\documentclass[twocolumn]{autart}    
\usepackage{amsfonts}
\usepackage{xcolor}
\usepackage{graphicx}          
\usepackage{epstopdf}
\usepackage{amsmath}
\usepackage{amssymb}
\usepackage{bm}
\usepackage[numbers,sort&compress]{natbib}
\usepackage{algpseudocode,algorithm}
\usepackage{multirow}
\usepackage{subfigure}
\usepackage{enumerate}
\usepackage{textcomp}
\usepackage{mathrsfs} 
\usepackage{setspace}
\usepackage{nomencl}
\usepackage{threeparttable}
\usepackage{comment}
\usepackage{appendix}
\usepackage{wrapfig}
\usepackage[colorlinks=true, linkcolor=blue, urlcolor=red, citecolor=black]{hyperref}

\graphicspath{{Graphics/}}
\newtheorem{lemma}{\bf Lemma}

\newtheorem{theorem}{\bf Theorem}
\newtheorem{remark}{\bf Remark}

\newtheorem{property}{Property}
\newtheorem{problem}{Problem}

\begin{document}
	\setlength{\abovedisplayskip}{3pt}
	\setlength{\belowdisplayskip}{3pt}
	\begin{frontmatter}
		
		\title{Variational Robust Kalman Filters: A Unified Framework}
		
		\thanks[footnoteinfo]{Shilei Li, Dawei Shi, and Hao Yu are with School of Automation, Beijing Institute of Technology, Beijing 100081, China (e-mail: shileili@bit.edu.cn, yuhaocsc@bit.edu.cn, daweishi@bit.edu.cn). Ling Shi is with the Department of Electronic and Computer Engineering, The Hong Kong University of Science and Technology, Kowloon, Hong Kong (e-mail: eesling@ust.hk). (Corresponding author: Dawei Shi.)}
		\author[ad1]{Shilei Li}\ead{shileili@bit.edu.cn}, 
		\author[ad1]{Dawei Shi}\ead{daweishi@bit.edu.cn},
		\author[ad1]{Hao Yu}\ead{yuhaocsc@bit.edu.cn},
		\author[ad2]{Ling Shi}\ead{eesling@ust.hk}
		
		\address[ad1]{School of Automation, Beijing Institute of Technology, Beijing 100081, China}
		\address[ad2]{Department of Electronic and Computer Engineering, The Hong Kong University of Science and Technology, Clear Water Bay, Kowloon, Hong Kong, China}
		
		
		\begin{abstract}
			Robustness and adaptivity are two competing objectives in Kalman filters (KF). Robustness involves temporarily inflating prior estimates of noise covariances, while adaptivity updates prior beliefs by exploiting measurements. In practical applications, both process and measurement noise can be influenced by outliers, be time-varying, or both. In this work, we propose a variational robust Kalman filter, built on a Student's $t$-distribution induced loss function and variational inference, and solved in a computationally efficient manner. We demonstrate that robustness can be understood as a prerequisite for adaptivity, making it possible to merge the above two competing goals into a single framework through a probabilistic switching rule. Additionally, our proposed filter can recover conventional KF, robust KF, and adaptive KF by tuning parameters, and can suppress both the imperfect process and measurement noise, enabling it to perform superiorly in complex noise environments. Simulations verify the effectiveness of the proposed method.
		\end{abstract}
	\end{frontmatter}
	
	\section{Introduction}
	State estimation aims to infer the latent states of a system from a mathematical model and noisy measurements, serving as a critical foundation for numerous fields, including robotics, finance, medical imaging, and meteorology~\cite{n2,n4}. The Kalman filter, co-developed by Kalman and Bucy, is a cornerstone in state estimation, providing optimal estimates in the mean squared error sense. Despite its theoretical elegance, the Kalman filter's reliance on Gaussian noise assumptions limits its efficacy in complex systems subject to imprecise noise covariance or unknown noise models~\cite{n7}.
	
	Research on filtering with non-Gaussian noise or unknown covariance is active. There are two main lines of research efforts: robustness and adaptability. The vulnerability of the Kalman filter to outliers has been found in its early history~\cite{n10}. To enhance its robustness, the $H_{\infty}$ filter, sometimes called minimax filter, was derived under the game theory framework~\cite{n11}. To combat gross errors,  some M-estimation-based estimators were developed, e.g.,  the Huber-based filter~\cite{n14}, $\ell_1$ norm-based Kalman filter~\cite{n15}. An alternative approach to design robust estimators is to exploit heavy-tailed distributions, including the Laplace-based filter~\cite{n16} and Student's $t$-based filter~\cite{n17}. Recently, the correntropy and statistical similarity measure, given their roots in the information-theoretic learning~\cite{principe2010information} and statistical learning~\cite{vapnik2013nature}, were utilized for algorithm derivation and achieving satisfactory results, including maximum correntropy Kalman filter (MCKF)~\cite{n18,li2023generalized} and statistical similarity measure based Kalman filter (SSMKF)~\cite{b12}. Both MCKF and SSMKF utilized a fixed-point iteration in optimization, which exhibited a higher computation efficiency than the gradient descent-based algorithms~\cite{kulikova2019chandrasekhar,singh2009using}. Current fixed-point convergence proofs are limited to specific robust losses. In this work, we expand this scope by providing sufficient conditions (Theorem \ref{theorem3}).
	
	Another prominent line is adaptive filtering, which can be broadly categorized into Bayesian, maximum likelihood \cite{n21}, covariance matching \cite{n22}, and variational Bayesian (VB) methods. Among these, Bayesian approaches are the most general, exemplified by techniques such as the interacting multiple model Kalman filter (IMMKF) and particle filters. However, these methods may suffer from high computational complexity. In contrast, VB methods provide a more computationally efficient alternative by approximating posterior inference through conjugate distributions~\cite{Blei03042017,b6}. A pioneering work in this area was presented by Särkkä~\cite{b6}, focusing on adaptive measurement covariance. This was later extended to include adaptive process and measurement covariance~\cite{b7,Huang2021VariationalAK}, along with nonlinear counterparts such as the unscented and cubature Kalman filters~\cite{n24,n25}. Although many variants of variational Bayesian Kalman filter (VBKF) had been developed, one of its key issues, such as the covariance tracking speed and the corresponding tracking variance, have not been explored. In this work, we provide the corresponding covariance tracking speed and reveal an inherent trade-off between convergence speed and steady state variance regarding the forgetting factor.    
	
	Many previous works concentrate on either robustness \cite{li2023generalized,b5,b9} or adaptability~\cite{b6,b8}. It is worth clarifying that the terminology ``robust adaptive" mentioned in \cite{b9} actually referred to adaptive robustness and hence falls into the first class. The VBKF~\cite{b6,b8}, although adjusting both shape and scale hyperparameters, falls into the second class. Only a limited number of studies consider the coexistence of time-varying covariance and outliers~\cite{chang2014kalman,li2016variational, Gao2021CubatureKF, Zhu2021AnAK}, with solutions including a combination of VB and robust loss techniques~\cite{chang2014kalman,li2016variational, Gao2021CubatureKF} or making joint inference on the noise types, the covariance, and the state~\cite{Zhu2021AnAK}. For instance, Chang \cite{chang2014kalman} proposed an innovation-based measure to switch between the standard Kalman filter, the fading memory Kalman filter, and the robust Kalman filter to balance between robustness and adaptability. Li et al. \cite{li2016variational} addressed this problem by designing a variational Huber-based filter by combining M-estimation with VB approximation. Gao et al.~\cite{Gao2021CubatureKF} fuses the ``robustness" and ``adaptability'' through the interacting multiple model. Zhu et al. \cite{Zhu2021AnAK} proposed an adaptive Kalman filter to handle inaccurate noise covariances and outliers by modeling the noise with Gaussian-Gamma mixture distributions, thereby jointly inferring the state, noise covariances, and hidden indicator variables within a variational Bayesian framework. 
	
	In this work, we demonstrate that the robust filter can be understood as a prerequisite of adaptive filters, allowing us to solve the robust adaptive problem in an efficient way. This avoids executing the ``robustness" and ``adaptability" in parallel as is done in ~\cite{chang2014kalman,li2016variational, Gao2021CubatureKF}, or simultaneously inferring both the noise types and the unknown covariances. Our developed filters can recover the KF, the robust KF, and the VBKF easily, formulating a unified framework. The contributions of this paper are summarized as follows.
	\begin{itemize}
		\item We provide a sufficient condition for the convergence of the fixed-point iteration under a class of robust losses in \textbf{Theorem \ref{theorem3}}. This allows us to derive a series of robust filters, including the Student’s $t$-based Kalman filter (STKF, \textbf{Algorithm \ref{AlgSTKF}}).
		\item We demonstrate that variational inference on the state vector with a point posterior distribution can be regarded as maximum a posteriori (MAP) estimation (\textbf{Theorem \ref{theoremiden}}), which inherently formulates a robust filter. Subsequently, we recover the state error covariance through the Joseph form update. This pipeline is much more computationally efficient than the conventional VB~\cite{b6}.
		\item We demonstrate that the proposed estimator can seamlessly recover the standard KF, the robust KF, the adaptive KF, or any hybrid combination through hyperparameter tuning. This provides a unified framework for versatile applications, whose effectiveness is verified through extensive simulations.
	\end{itemize}
	
	The remainder of this paper is organized as follows. Section II introduces some preliminaries.  Section III derives a family of robust filters with guaranteed convergence. Section IV develops some robust adaptive filters. Section V provides some simulations. Section VI concludes the paper.
	
	\emph{Notations}: The transpose of a matrix $A$ is denoted by $A^{T}$. $X\succ 0$ ($X \succcurlyeq 0$) denotes $X$ is positive definite (semi-positive definite) matrix. The Gaussian distribution with mean $\mu$ and covariance $\Sigma$ is denoted by $\mathcal{N}(\mu,\Sigma)$. The expectation of a random variable $X$ or random vector $\mathcal{X}$ is denoted by $\operatorname{E}(X)$ or $\operatorname{E}(\mathcal{X})$. The prior and posterior estimate of state $x$ is denoted as $x^{-}$ and $x^{+}$, respectively.
	\section{Preliminaries}
	We begin by presenting the necessary preliminaries and lemmas, followed by the problem statement.
	\subsection{The Student's $t$-distribution}
	The probability distribution functions (PDF) of Student's $t$-distribution, Gaussian distribution, and inverse Gamma distribution of random variable $X$ are given as follows:
	\begin{subequations}
		\tiny
		\begin{align}
			\operatorname{St}\left(x \mid \nu, \mu, \tau^2\right)&=\frac{\Gamma\left(\frac{\nu+1}{2}\right)}{\Gamma\left(\frac{\nu}{2}\right) \sqrt{\pi \nu \tau^2}}\left(1+\frac{1}{\nu} \frac{(x-\mu)^2}{\tau^2}\right)^{-(\nu+1) / 2}, \label{lst}\\
			\mathcal{N}(x|\mu,\tau^2)&=\frac{1}{\sqrt{2\pi}\tau}\exp\Big(-\frac{1}{2}\frac{(x-\mu)^2}{\tau^2}\Big),\label{lgau}\\
			\operatorname{Inv-Gam}(x|a,b)&=\frac{b^a}{\Gamma(a)}(1 / x)^{a+1} \exp (-b / x),
		\end{align}
		\label{pdf}
	\end{subequations}
	where $\operatorname{St}\left(x \mid \nu, \mu, \tau^2\right)$ denotes the Student's $t$-distribution with degrees of freedom (DOF) $\nu$, mean $\mu$, and scale parameter $\tau^2$, $\mathcal{N}(x|\mu,\tau^2)$ denotes the Gaussian distribution with mean $\mu$ and covariance $\tau^2$, and $\operatorname{Inv-Gam}(x|a,b)$ denotes the inverse-Gamma distribution with shape parameter $a$ and scale parameter $b$~\cite{gelman1995bayesian}.  
	\begin{property}[\cite{gelman1995bayesian}]
		\label{p1}
		The Student's $t$-distribution $\operatorname{St}(x|\nu,\\ \mu, \tau^2)$ is a compounding distribution composed of a Gaussian distribution and Inverse-Gamma distribution, i.e., it is equivalent to 
		\begin{equation}
			x\sim \mathcal{N}(\mu,\lambda),
		\end{equation}
		where $\lambda$ follows
		\begin{equation}
			\lambda \sim \operatorname{Inv-Gam}(\lambda| \frac{\nu}{2},\frac{\nu \tau^2}{2}).
			\label{plambda}
		\end{equation}
	\end{property} 
	
	Denoting $e=x-\mu$ and taking the logarithm on the right hand side of \eqref{lst} and ignoring the constant terms, we obtain the Student's $t$-induced loss $\mathcal{L}_{st}^{*}$ as follows
	\begin{equation}
		\mathcal{L}_{st}^{*}=  \frac{\nu+1}{2}\log\left(1+ \frac{e^2}{\nu\tau^2}\right).
		\label{InducedLoss}
	\end{equation}
	
	\begin{lemma}[Invariant Loss]
		\label{lemma:invariant_loss}
		Let $e = [e_1, \ldots, e_l]^T \in \mathbb{R}^l$ be an error vector with statistically independent components. Minimizing the following two loss functions yields identical optimal estimates:
		\begin{equation}
			\begin{aligned}
				\mathcal{L}_{st}^{*} &= \sum_{i=1}^{l} \frac{\nu_i+c}{2} \log\left(1+ \frac{e_i^2}{\nu_i\tau_i^2}\right), \\ 
				\mathcal{L}_{st} &= \sum_{i=1}^{l} \frac{\nu_i}{2} \log\left(1+ \frac{e_i^2}{\nu_i\tau_i^2}\right)
			\end{aligned}
		\end{equation}
		where $c$ is a constant, and $\nu_i$ are either all identical or partially tend to infinity.
	\end{lemma}
	
	The lemma can be proved by the separability of the objective function across independent components $e_i$ and the scale invariance of each decoupled term.
	
	By analogy with the obtainment of $\mathcal{L}_{st}^{*}$, we obtain the Gaussian-induced loss $\mathcal{L}_{gau}$ according to \eqref{lgau}. We compare $\mathcal{L}_{st}$ and $\mathcal{L}_{gau}$ as follows
	\begin{equation}
		\begin{aligned}
			\mathcal{L}_{st}&=  \frac{\nu}{2}\log\left(1+ \frac{e^2}{\nu\tau^2}\right),~\mathcal{L}_{gau} =  \frac{1}{2}\frac{e^2}{\tau^2}.
		\end{aligned}
		\label{loss1}
	\end{equation}
	Its influence functions, measuring the derivative of the loss with respect to the error, are given as 
	\begin{equation}
		\begin{aligned}
			\frac{\partial \mathcal{L}_{st}}{\partial e} &= \frac{\nu e}{\nu \tau^2+{e^2}},~\frac{\partial \mathcal{L}_{gau}}{\partial e}=\frac{e}{\tau^2}.
		\end{aligned}
	\end{equation}
	\begin{property}
		\label{p2}
		As $\nu \to \infty$, one has $\mathcal{L}_{st}=\mathcal{L}_{gau}$. 	As $\nu=1$, $\mathcal{L}_{st}$ becomes Cauchy loss corresponding to Cauchy distribution.
	\end{property}
	The proof is available in Appendix \ref{proofproperty2}.
	\begin{property}
		\label{p3}
		The loss $\mathcal{L}_{st}$ is convex within the region $[-\sqrt{\nu}\tau,\sqrt{\nu}\tau]$, and is concave in other regions. 
	\end{property}
	The proof is available at Appendix \ref{proofproperty3}. 
	\begin{property}
		\label{p4}
		As $\nu \to \infty$, the PDF of latent variable $\lambda$ (see \eqref{plambda}) becomes a shifted Dirac delta function $\delta(\lambda-\tau^2)$.
	\end{property}
	The proof is available at Appendix \ref{proofproperty4}. 
	
	\begin{figure}[!htp]
		\centering
		\subfigure[loss function]{
			\begin{minipage}[t]{0.45\linewidth}
				\centering
				\includegraphics[width=1.0\columnwidth]{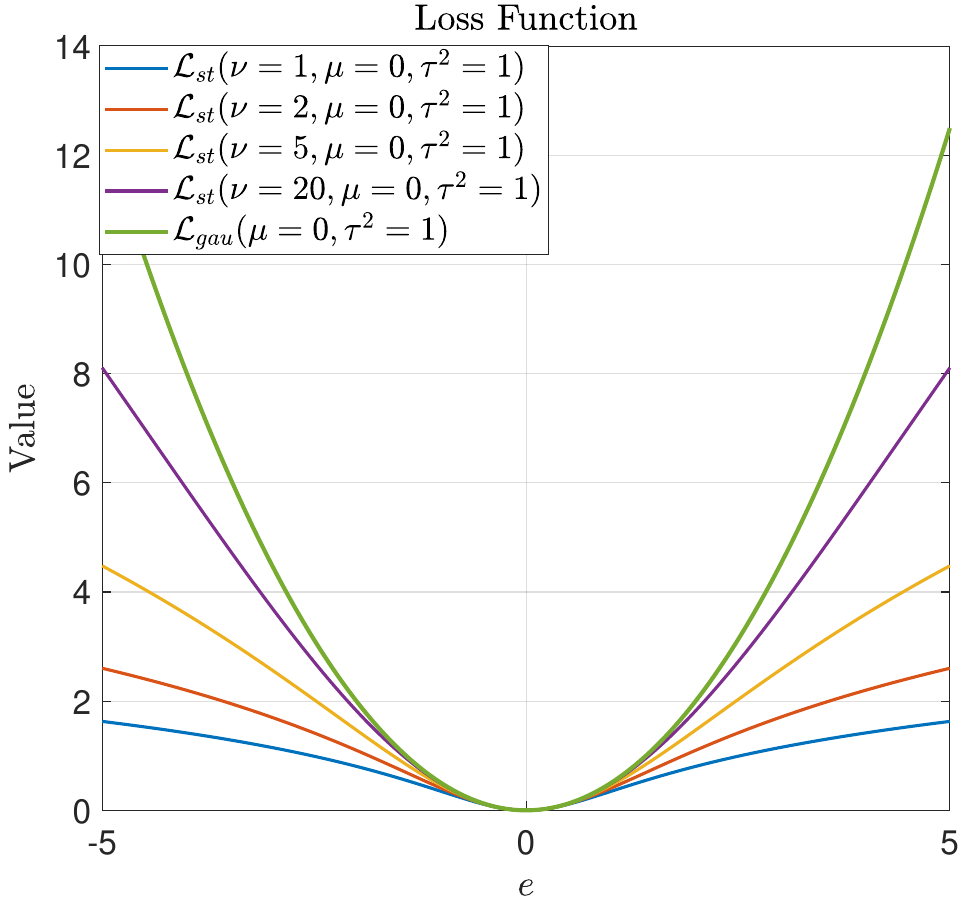}
				\label{loss}
			\end{minipage}%
		}%
		\subfigure[influence function]{
			\begin{minipage}[t]{0.45\linewidth}
				\centering
				\includegraphics[width=1.0\columnwidth]{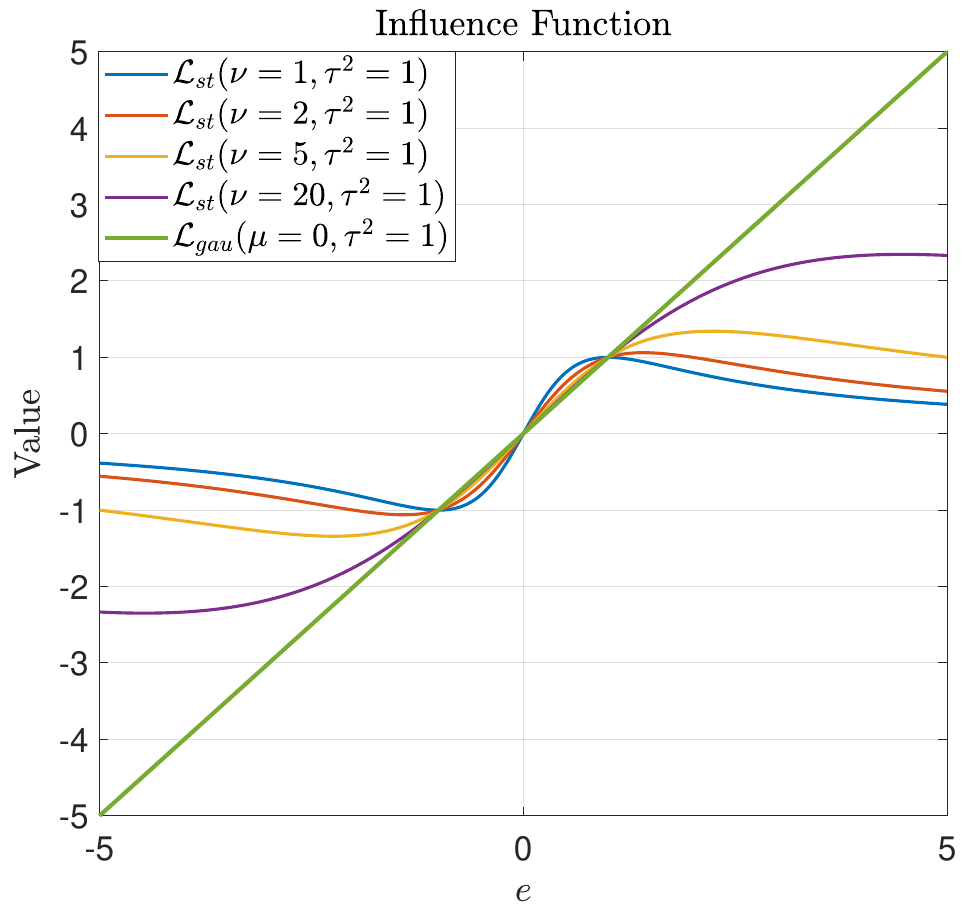}
				\label{inf}
			\end{minipage}%
		}%
		\\
		\subfigure[PDF]{
			\begin{minipage}[t]{0.45\linewidth}
				\centering
				\includegraphics[width=1.0\columnwidth]{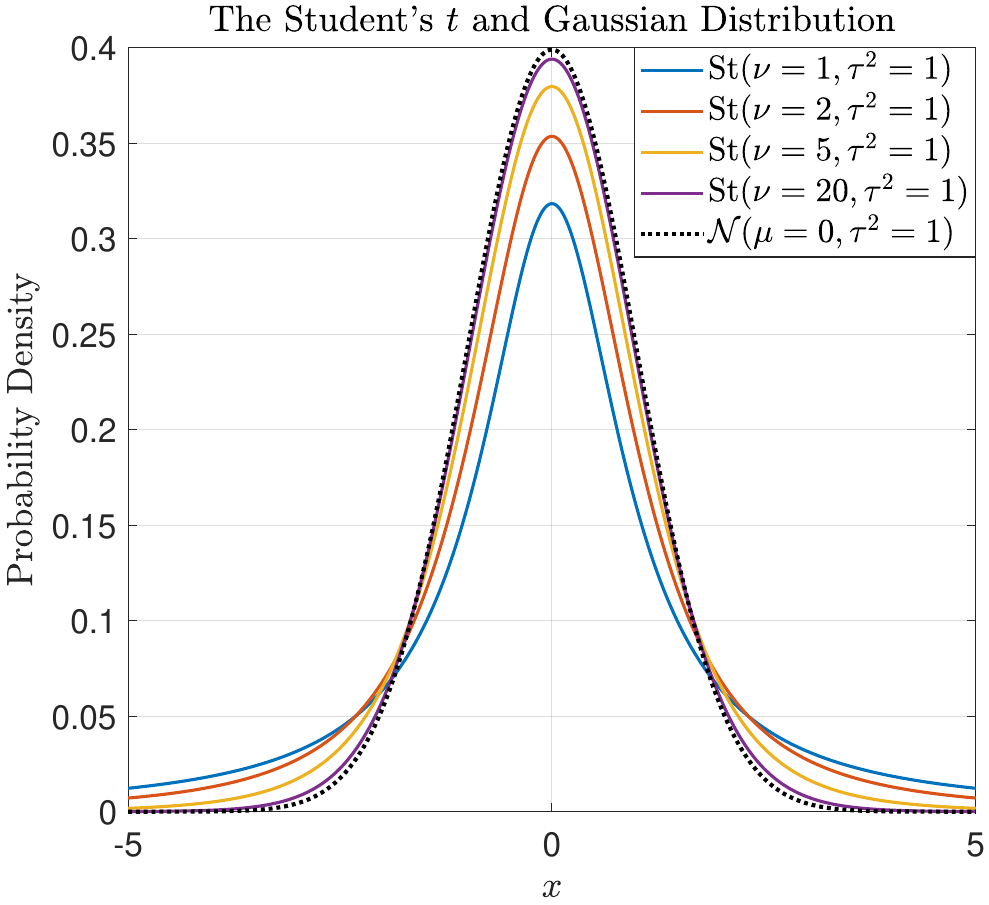}
				\label{student}
			\end{minipage}%
		}%
		\subfigure[Latent PDF]{
			\begin{minipage}[t]{0.45\linewidth}
				\centering
				\includegraphics[width=1.0\columnwidth]{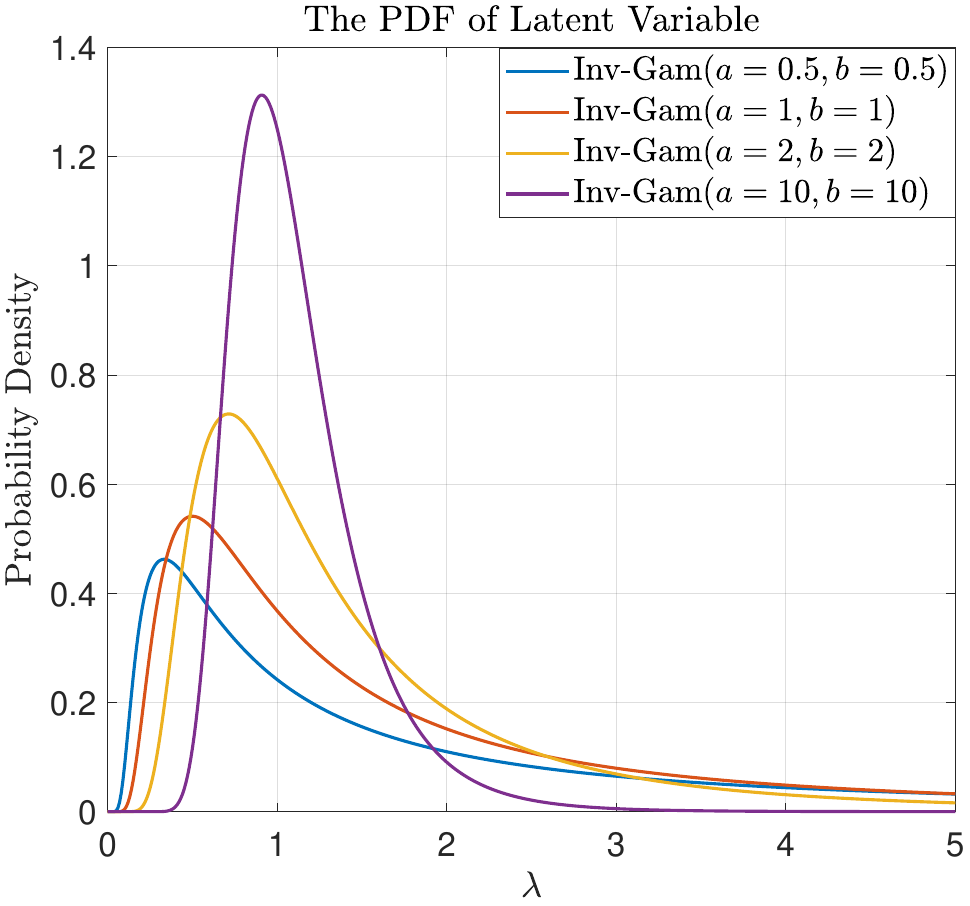}
				\label{inverse}
			\end{minipage}%
		}%
		\caption{The visualization of $	\mathcal{L}_{st}$ and $	\mathcal{L}_{gau}$ as well as their influence functions and induced PDFs. (a) The loss function of $\mathcal{L}_{st}$ and $\mathcal{L}_{gau}$. (b) The influence function of $\mathcal{L}_{st}$ and $\mathcal{L}_{gau}$. (c) The mapped Student's t distribution and Gaussian distribution. (d) The PDF of latent variable $\lambda$.}
		\label{property}	
	\end{figure}
	The loss and influence functions of $\mathcal{L}_{st}$ and $\mathcal{L}_{gau}$, as well as their corresponding PDFs and the induced latent variable's PDF in $\mathcal{L}_{st}$ are visualized in Fig. \ref{property}. One observes that the influence function of $\mathcal{L}_{st}$ has a redescending property, indicating its robustness to absolute errors that are much greater than $\sqrt{\nu}\tau$. According to Properties \ref{p1}, \ref{p2}, \ref{p3}, and \ref{p4}, we have the following insights:
	\begin{itemize}
		\label{inversegammaproperty}
		\item The parameter $\tau^2$ describes the variance of the latent Gaussian distribution.
		\item The parameter $\nu$ reflects how confident we are about the latent variance $\tau^2$.  
	\end{itemize}
	We leverage these insights in the subsequent section.
	\subsection{Problem Statement}
	\label{pd}
	We consider the following linear system:
	\begin{equation}
		\begin{aligned}
			x_{k+1}=A x_{k} + w_k \\
			y_k = C x_{k} +v_k
		\end{aligned}
		\label{sys}
	\end{equation}
	where $A \in \mathbb{R}^{n\times n}$ and $C \in \mathbb{R}^{m \times n}$ are state transfer and observation matrices. The process noise $w_k$ and measurement noise $v_k$ are uncorrelated zero-mean random noises. The initial state $x_0$ is zero-mean Gaussian with known covariance matrix $P_0$, and is independent of $w_k$ and $v_k$ for all $k>0$. We consider that the noise covariance may be time-varying and occasionally contaminated by outliers, i.e., 
	\begin{equation}
		\begin{aligned}
			w_k \sim \epsilon_w p(w_{n,k}) + (1-\epsilon_w)p_{ol}(o_w),\\
			v_k\sim \epsilon_v p(v_{n,k}) + (1-\epsilon_v)p_{ol}(o_v),
			\label{probstate}
		\end{aligned}
	\end{equation}
	where $w_{n,k}\sim \mathcal{N}(0,Q_k)$ and $v_{n,k}\sim \mathcal{N}(0,R_k)$. The parameters $\epsilon_w, \epsilon_v \in (0, 1]$ denote the occurrence probabilities of the nominal noises, while $p_{w,ol}(\cdot)$ and $p_{v,ol}(\cdot)$ denote the unknown probability distributions of the outliers. The model \eqref{probstate} recovers the standard KF when $Q_k = Q$, $R_k = R$, and $\epsilon_w = \epsilon_v = 1$. It simplifies to the robust filtering problem when $Q_k \triangleq Q$, $R_k \triangleq R$ with $\epsilon_w, \epsilon_v < 1$, and it becomes the adaptive filtering problem when $\epsilon_w = \epsilon_v = 1$. 
	
	The objective of this work is to develop a unified recursive estimator capable of jointly estimating the system state $x_k$ and the time-varying nominal covariances $Q_k$ and $R_k$, while simultaneously maintaining robustness against occasional outliers.
	\section{Main Results}
	In this section, we first present a robust filter under $\mathcal{L}_{st}$ with proved convergence. Next, we establish a connection between the robust filter and the adaptive filter under the Kullback–Leibler divergence perspective. Finally, we update the hyperparameters using variational inference and design a probabilistic switching rule for scenarios where outliers and time-varying noise coexist.
	\subsection{A Class of Robust Filters}
	By regarding the prior estimate of the state $x_k^{-} = A x_{k-1}^{+}$ as a pseudo-measurement, we rewrite \eqref{sys} as
	\begin{equation}
		\begin{aligned}
			&t_k=W_k x_k+\zeta_k,
		\end{aligned}
		\label{linreg}
	\end{equation}
	with 
	\begin{equation}\nonumber
		t_k=B_k^{-1}\begin{bmatrix}
			{x_k^{-}}\\
			{y_k}
		\end{bmatrix},\quad W_k=B_k^{-1}\begin{bmatrix}
			I\\
			C
		\end{bmatrix}\\, \zeta_k=B_k^{-1}\begin{bmatrix}
			{e_{x,k}}\\
			{v_k}
		\end{bmatrix}
	\end{equation}
	where $B_kB_k^{T}=\begin{bmatrix}
		P_{k}^{-}&0\\
		0&R_k
	\end{bmatrix}$, $P_{k}^{-}$ denotes the prior error covariance, and $\zeta_k$ is the normalized measurement noise.	
	\begin{remark}
		\label{remarkmap}
		\textcolor{black}{By constructing the augmented pseudo-measurement vector $t_k$, the prior state estimate $x_k^{-}$ is embedded directly into the observation space. Statistically, solving this augmented linear regression is mathematically equivalent to maximizing the posterior distribution $p(x_k \mid y_{1:k}) \propto p(y_k \mid x_k)p(x_k \mid y_{1:k-1})$. This unifies the Bayesian prediction and update phases into a single optimization framework.}
	\end{remark}
	
	Instead of minimizing the least squares loss $\|e_k\|_{2}^{2}$ where $e_k=t_k-W_k x_k \in \mathbb{R}^{l}$ wiht $l=n+m$ as done in KF~\cite{li2023generalized}, by denoting $R_{\zeta\zeta_k}=\mathrm{E}(\zeta_k \zeta_k^{T})=\operatorname{diag}(\lambda_1,\ldots,\lambda_i,\cdots, \lambda_{l})$ and formulating $\lambda_i$ as an inverse gamma distribution, we construct the following MAP problem:
	\begin{equation}
		\begin{aligned}
			p&(x_k, R_{\zeta\zeta_k} \mid t_{1:k}) = \frac{p(x_k, R_{\zeta\zeta_k}, t_k \mid t_{1:k-1})}{p(t_k \mid t_{1:k-1})} \\
			&\propto p(t_k, \chi_k \mid R_{\zeta\zeta_k}) p(R_{\zeta\zeta_k})\\
			&=\mathcal{N}(t_k;W_k x_k, R_{\zeta\zeta_k})\prod_{i=1}^{n+m}\operatorname{IG}(\lambda_i| \frac{\nu_i}{2},\frac{\nu_i \tau_i^2}{2}).
		\end{aligned}
		\label{map}
	\end{equation}
	According to the composition property of Student's $t$-distribution and equation \eqref{InducedLoss},  taking negative logarithm on both sides of \eqref{map} and ignoring the constant term gives 
	\begin{equation}
		J=\sum_{i=1}^{l} \frac{\nu_i+1}{2} \log(1 + (e_{i,k}^{2}) / (\nu_i \cdot \tau_i^2)),
		\label{rloss}
	\end{equation}
	where $e_{i,k}=t_{i,k}-w_{i,k}x_k$. According to Lemma \ref{lemma:invariant_loss}, we reformulate \eqref{rloss} as 
	\begin{equation}
		J=\sum_{i=1}^{l} \frac{\nu_i}{2} \log(1 + (e_{i,k}^{2}) / (\nu_i \cdot \tau_i^2)).
		\label{rloss1}
	\end{equation}
	By letting $\frac{\partial J}{\partial x_k} =0$, one has
	\begin{equation}
		x_k=\big(\sum_{i=1}^l \frac{\nu_i w_i^T w_i}{\nu_i \tau_i^2+e_{i, k}^{2}}\big)^{-1}\big({\sum_{i=1}^l \frac{\nu_i w_i^T t_i}{\nu_i \tau_i^2+e_{i, k}^{2}}}\big).
		\label{stloss}
	\end{equation}
	It follows
	\begin{equation}
		x_k = f(x_k)=(W_k^{T}D_k W_k)^{-1} W_k^{T} D_k t_k,
		\label{fixedpoint}
	\end{equation}
	where $D_k=\operatorname{diag}([d_{\nu_1},d_{\nu_2},\ldots,d_{\nu_l}])$ and  
	\begin{equation}
		d_{\nu_i}(e_{i,k})=\frac{\nu_i}{\nu_i \tau_i^{2}+e_{i,k}^{2}}.
		\label{dnu}
	\end{equation}
	Note that the equation \eqref{fixedpoint} can be regarded as the least square solution of \eqref{linreg} where $E(\zeta_k^{T}\zeta_k)=D_k^{-1}=\operatorname{diag}([\lambda_1,\lambda_2,\ldots,\lambda_l])$ and
	\begin{equation}
		\lambda_i \triangleq d^{-1}_{\nu_i} = \underbrace{\tau_i^{2}}_{\text{constant}} + \underbrace{\frac{e_{i,k}^{2}}{\nu_i}}_{\text{inflation}}.
		\label{inflationeffect}
	\end{equation}
	We observe that both sides of \eqref{fixedpoint} contain $x_k$, as $e_k=t_k -W_k x_k$. Then, it is natural to apply a fixed-point iteration for the solution of \eqref{fixedpoint}. Using a similar derivation as done in \cite{li2023generalized}, we obtain the Student's $t$-based Kalman filter (STKF) summarized in Algorithm \ref{AlgSTKF}. To provide the convergence analysis of Line 10-17 in Algorithm \ref{AlgSTKF}, we provide the following lemma.
	\begin{algorithm}[bt]
		\setstretch{1.0} 
		\caption{STKF}
		\label{AlgSTKF}
		\begin{algorithmic}[1]
			\State \textbf{Step 1: Initialization}\\
			Choose $\nu_i$ and $\tau_i^2$ for channel $i$, maximum iteration number $m_\mathrm{iter}$, and a threshold $\varepsilon$.
			\State {\textbf{Step 2: State Prediction}}\\
			$\hat{{x}}_{k}^{-}=A \hat{{x}}_{k-1}^{+} $\\
			${P}_{ k}^{-}={A} {P}_{k-1 }^{+} {A}^{T}+{Q}_k$\\
			Obtain ${B}_{p}$ and ${B}_{r}$ with ${B}_{p}{B}_{p}^{T}={P}_{k}^{-}$ and ${B}_{r}{B}_{r}^{T}={R}_{k}^{*}$\\
			Obtain $t_k$ and $W_k$ through \eqref{linreg}
			\State \textbf{Step 3: State Update}\\
			$\hat{{x}}_{k,0}^{+}=\hat{{x}}_{k}^{-}$
			\While{$\frac{\left\|\hat{{x}}_{k,t}^{+}-\hat{{x}}_{k,t-1}^{+}\right\|}{\left\|\hat{{x}}_{k,t}^{+}\right\|}>\varepsilon$ or $t \le m_\mathrm{iter}$} \\
			$\hat{{x}}_{k,t}^{+}=\hat{{x}}_{k}^{-}+\tilde{K}_{k,t}({y}_k-{C}\hat{{x}}_{k}^{-})$ \Comment{$t$ starts from 1}\\
			$\tilde{{K}}_{k,t}=\tilde{P}_{ k}^{-}{C}^{T}(C\tilde{P}_{ k}^{-}{C}^{T}+\tilde{R}_{k})^{-1}$\\
			$\tilde{{P}}_{ k}^{-}={B}_{p}{{M}}_{p}^{-1}{{B}}_{p}^{T}$\\
			${M}_{p}=\operatorname{diag}(d_{\nu_i}(e_{1,k}),\ldots,d_{\nu_i}(e_{n,k}))$\\
			$\tilde{R}_{ k}={B}_{r}{{M}}_{r}^{-1}{B}_{r}^{T}$\\
			${M}_{r}=\operatorname{diag}(d_{\nu_i}(e_{n+1,k}),\ldots,d_{\nu_i}(e_{n+m,k}))$\\
			${e}_{i,k}=t_{i,k}-w_{i,k}x_{k,t-1}^{+}$ \Comment{$t_{i,k}$ is $i$-th element of $t_k$}\\
			$t=t+1$
			\EndWhile\\
			${P}_{k}^{+}=({I}-\tilde{{K}}_k {C}){{P}}_{ k}^{-}({I}-\tilde{{K}}_k {C})^{T}+\tilde{{K}}_k {R}_{k}\tilde{{K}}_k^{T}$		
		\end{algorithmic}
	\end{algorithm}  
	\begin{lemma}[\cite{chen2015convergence}]
		\label{convlemma}
		The fixed-point	algorithm \eqref{fixedpoint} converges if $\exists~ \gamma>0$ and $0<\eta<1$ such that the initial vector $\|x_{0}\|_p < \gamma$, and $\forall x \in \{x \in \mathbb{R}^{n}: \|x\|_p \le \gamma\}$, the following holds (omitting time index $k$ for simplicity)
		\begin{equation}
			\begin{aligned}
				\left\{\begin{array}{l}
					\|f(x)\|_p \le \gamma,\\ 
					\| \nabla_{x} f(x) \|_p \le \eta,
				\end{array}\right.
			\end{aligned}
			\label{convergence}
		\end{equation}
		where $\|\cdot \|_p$ denotes an $\ell_p$ norm of a vector or an induced norm of a matrix defined by $\|A\|_p = \max \limits_{\|x\|_p} \frac{\|Ax\|_p}{\|x\|_p}$ with $p \ge 1$, and $\nabla_x f(x)$ is the Jacobian matrix of $f(x)$.
	\end{lemma}
	Denote the DOF vector as ${\boldsymbol{\nu}}=[\nu_1, \nu_2, \ldots, \nu_l]^{T} \in \mathbb{R}^{l}$ and the unified DOF as $\nu_1=\nu_2=\cdots=\nu_l=\nu \in \mathbb{R}$. The following Theorem \ref{theorem1} ensures the validity of first inequality in \eqref{convergence}, while Theorem \ref{theorem2} guarantees the second one.
	\begin{theorem}
		Let $f(x)$ be defined as in \eqref{fixedpoint}. If $\gamma > \xi$ and $\nu_i \ge \nu^*$ for all $i=1,\dots,l$, then $\|f(x)\|_1 \le \gamma$ for every $x \in \mathbb{R}^n$ satisfying $\|x\|_1 \le \gamma$. Here, $\xi$ is defined as
		\begin{equation}
			\xi = \frac{\sqrt{n} \sum_{i=1}^{l}\frac{1}{\tau_i^2}|t_i|\|w_i^{T}\|_1}{\lambda_{\min}\left[\sum_{i=1}^{l}\frac{1}{\tau_i^2} w_i^{T}w_i\right]}, 
			\label{xixi}
		\end{equation} and $\nu^*$ is the solution to $\phi(\nu)=\gamma$, where
		\begin{equation}
		\phi(\nu)= \frac{\sqrt{n}\sum_{i=1}^{l}|t_i|\|w_i^{T}\|_1}{\lambda_{\min}\left[\sum_{i=1}^{l} d_{\nu}(\gamma \|w_i^{T}\|_1+t_i)w_i^{T}w_i\right]}. 
		\label{con1}
		\end{equation}
		\label{theorem1}
	\end{theorem}
	The proof is available at \ref{pftheorem1}.
	\begin{theorem}
		\label{theorem2}
		Let $f(x)$ be defined as in \eqref{fixedpoint}, $\forall x \in \{x \in \mathbb{R}^{n}: \|x\|_1 \le \gamma\}$, it holds that $\|f(x)\|_1 \le \gamma$, and $\|\nabla_{x} f(x)\| \le \eta$ if 
		$\gamma > \xi$ (see \eqref{xixi}), where $\nu^{*}$ is the solution of $\phi(\nu)=\gamma$, and $\nu^{+}$ is the solution of $\psi(\nu)=\eta ~(0<\eta<1)$ with
		\begin{equation}
			\tiny
			\begin{aligned}
				{\psi}({\nu})= \frac{2\sqrt{n} \sum_{i=1}^{l}\frac{|t_i|+\gamma\|w_i^{T}\|_{1}}{\tau_i^{4}} \|w_i^{T}\|_1  \big{(}\gamma\|w_i^{T}w_i \|_1+ \|w_i^{T}t_i\|_1 \big{)}}{ \nu\lambda_{\min}[\sum_{i=1}^{l}w_i^{T}d_{\nu}(|t_i|+\gamma |w_i^{T}|_1) w_i]}.
				\label{con2}
			\end{aligned}
		\end{equation}
	\end{theorem}	
	The proof is available at \ref{pftheorem2}.
	\begin{remark}
		Theorems \ref{theorem1} and \ref{theorem2} provide a sufficient condition for the convergence of the fixed-point iteration \eqref{fixedpoint} under dynamics \eqref{linreg} and loss \eqref{rloss1}. 
	\end{remark}
	
	We further extend Theorems \ref{theorem1} and \ref{theorem2} to a general loss $J_{\boldsymbol{\nu}}(e;\boldsymbol{\nu},\boldsymbol{\tau}^2)=\sum_{i=1}^{l}J_{\nu_i}(e_i;\nu_i,\tau_i^2)$ where $\nu_i>0$ and $\tau_i^2>0$ are the hyperparameters and $e_i \in \mathbb{R}$ is the $i$-th element of $e$. By denoting $\frac{\partial J_{\nu_i}}{\partial e_i} = d_{\nu_i}(e_i)e_i$ and $\frac{\partial d_{\nu_i}(e_i)}{\partial e_i} = \iota_{\nu_i}(e_i) e_i$, Theorem \ref{theorem3} provide a necessary condition for the convergence of a general loss using a fixed-point solution.
	\begin{theorem}
		\label{theorem3}
		Considering dynamics \eqref{linreg}, a general loss $J_{\boldsymbol{\nu}}(e;\boldsymbol{\nu},\boldsymbol{\tau}^2)=\sum_{i=1}^{l}J_{\nu_i}(e_i;\nu_i,\tau_i^2)$ can be solved by a fixed-point algorithm with guaranteed convergence if the following conditions are fulfilled:
		\begin{itemize}
			\item Condition 1: $J_{\nu_i}(e_i)$ is a non-decreasing continuous function with respect to $|e_i|$ and gives its minimum at $e_i=0$.
			\item Condition 2: $\exists \kappa \in \mathbb{R}^{+}$ so that $0 \le d_{\nu_i}(e_i)\le \kappa$. Moreover, $d_{\nu_i}(e_i)$ is an increasing function of $\nu_i$ with  $\lim \limits_{\nu_i \to 0_{+}}  d_{\nu_i}(e_i) = 0$ for any $|e_i|>0$.
			\item Condition 3: $\iota_{\nu_i}(e_i)$ is bounded for any $e_i$.
			\item Condition 4: $\nu_i > \max \{\nu^{*}, \nu^{+}\}$ for $i=1,2,\ldots,l$, where $\nu^{*}$ and $\nu^{+}$
			are constants determined analogously to the expressions in \eqref{con1} and \eqref{con2}, respectively.
		\end{itemize}
	\end{theorem}
	The proof is available at \ref{pftheorem3}.
	Some exemplary losses $J_{\boldsymbol{\nu}}(e;\boldsymbol{\nu},\boldsymbol{\tau}^2)=\sum_{i=1}^{l}J_{\nu_i}(e_i;\nu_i,\tau_i^2)$ are listed in Table \ref{costfun}. For simplicity, we use $l=1$ so the subscript 
	$i$ is omitted. We find that the proposed STKF uses the logarithmic loss for algorithm derivation while the MKCKF~\cite{n18} and AORSE-sqrt algorithm~\cite{b12} utilize the exponential loss and square root loss as optimization objectives, respectively (these three estimators correspond to the first three losses in Table \ref{costfun}). It is worth mentioning that a different robust filter can be obtained by replacing the expression of $d_{\nu_i}(e_{i,k})$ Lines 13 and 14 with the expressions in Table \ref{costfun} (the second column) in Algorithm \ref{AlgSTKF}. 
	\begin{table*}[]
		\centering
		\caption{Some Robust Losses Fulfilling Theorem \ref{theorem3}.}
		\scalebox{0.9}{
			\begin{tabular}{cccc}
				\hline
				$J_{\nu}(e)$ & $d_{\nu}(e)$ & $\iota_{\nu}(e)$ & special cases\\
				\hline
				$\frac{\nu}{2} \log \left(1+\frac{e^2}{\nu \tau^2}\right), \nu \in (0,\infty)$ & $\frac{\nu}{\nu \tau^{2}+e^{2}}$               & $- \frac{2 \nu}{(\nu \tau^{2}+e^{2})^{2}}$  &$\lim_{\nu\to\infty}J_{\nu}(e)=\frac{e^2}{2\tau^2}$\\
				$\nu^2 \big(1-\exp(-\frac{e^2}{2\nu^2\tau^2})\big),\nu \in (0,\infty)$                                           & $\frac{1}{\tau^2}\exp(-\frac{e^2}{2\nu^2\tau^2})$          & $ - \frac{1}{\nu^2\tau^4} \exp(- \frac{e^2}{2\nu^2\tau^2})$&$\lim_{\nu\to\infty}J_{\nu}(e)=\frac{e^2}{2\tau^2}$\\
				$\frac{2-\nu}{\nu} \big(\big(\frac{e^2/\tau^2}{2-\nu}+1\big)^{\nu/2}-1\big),\nu \in (0,2)$ & $\frac{1}{\tau^2}(\frac{e^2/\tau^2}{2-\nu}+1)^{\nu/2-1}$               & $- \frac{1}{\tau^4}(\frac{e^2/\tau^2}{2-\nu}+1)^{\nu/2-2}$&$\lim_{\nu\to 2}J_{\nu}(e)=\frac{e^2}{2\tau^2}$\\ 
				$\sqrt{\nu(\nu+e^2/\tau^2)}-\nu,\nu \in (0,\infty)$&  $\frac{1}{\tau^2\sqrt{1+\frac{e^2}{\nu\tau^2}}}$  &  $-\frac{1}{\nu\tau^4}\left(1+\frac{e^2}{\nu\tau^2}\right)^{-3/2}$  &$\lim_{\nu\to\infty}J_{\nu}(e)=\frac{e^2}{2\tau^2}$ \\ 
				\hline                                         
		\end{tabular}}
		\label{costfun}
	\end{table*}
	\begin{prop}
		The STKF, MKCKF~\cite{n18}, and AORSE-sqrt~\cite{b12} are identical to the standard KF~\cite{barfoot2024state} as $\nu_i \to \infty$ for $i=1,2,\ldots,l$.
		\label{prop1} 
	\end{prop}
	The proof can be obtained by taking the limit as $\nu_i \to \infty$ in the scaling equations of the STKF, MKCKF, and AORSE-sqrt, verifying their exact algebraic reduction to the standard KF updates.
	
	\subsection{Connection with Variational Adaptive Filter}
	To build a connection between the proposed robust filter (indeed, it is a MAP estimator) and the adaptive filter, we formulate the following two problems.
	\begin{problem}[MAP Estimation]
		Given the augmented regression model \eqref{linreg} and fixed noise covariance $R_{\zeta\zeta_k}$, the STKF seeks the MAP estimate (see Remark \ref{remarkmap} and \eqref{map}) of the state $x_k$ by solving:
		\begin{equation}
			x_k^* = \arg \max_{x_k}\ln p\left(x_k, R_{\zeta\zeta_k} |t_k \right).
			\label{map1}
		\end{equation}
		the posterior error covariance is subsequently updated by the Joseph form:
		\begin{equation}
			{P}_{k}^{+}=({I}-\tilde{{K}}_k {C}){{P}}_{ k}^{-}({I}-\tilde{{K}}_k {C})^{T}+\tilde{{K}}_k {R}_{k}\tilde{{K}}_k^{T}.
			\label{Josephform}
		\end{equation}
		where $\tilde{{K}}_k$ is the associate gain by optimizing \eqref{map1}.
		\label{prob1}
	\end{problem}
	
	\begin{problem}[Variational Bayesian Inference]
		The conventional Variational Bayesian Kalman Filter (VBKF) approximates the true joint posterior by a factorized distribution $Q(x_k, R_{\zeta\zeta_k}) = Q_x(x_k)Q_R(R_{\zeta\zeta_k})$ that minimizes the Kullback-Leibler (KL) divergence \cite{b6}:
		\begin{equation}
			\mathcal{L}_{vb}=\mathrm{KL}\Big(Q_x(x_k)Q_R({R}_{\zeta\zeta_k}) \,\big\|\, p\left(x_k, R_{\zeta\zeta_k} \mid t_k\right)\Big).
		\end{equation}
		\label{prob2}
	\end{problem}
	\begin{theorem}
		\textcolor{black}{Given a fixed noise covariance $Q_R(R_{\zeta\zeta_k})$, such that the expected precision is $\mathbb{E}_{Q_R}[R_{\zeta\zeta_k}^{-1}] \triangleq \tilde{R}_{\zeta\zeta_k}^{-1}$. the optimal variational posterior $Q_x^*(x_k)$ solving Problem \ref{prob2} is a Gaussian whose mean is identical to the MAP state estimate derived in Problem \ref{prob1}, and the posterior error covariance is identical to \eqref{Josephform}.}
		\label{theoremiden}
	\end{theorem}
	\begin{pf}
		In the Variational Bayesian framework \cite{Weinstock1952CalculusOV}, minimizing the KL divergence is equivalent to maximizing the Evidence Lower Bound (ELBO), denoted as $\mathcal{L}$:
		\begin{equation}
			\ln p(t_k) = \mathcal{L} + \mathrm{KL}\Big(Q_x(x_k)Q_R(R_{\zeta\zeta_k}) \,\big\|\, p(x_k,R_{\zeta\zeta_k} \mid t_k)\Big),
		\end{equation}
		where the ELBO is defined as:
		\begin{equation}
			\begin{aligned}
				\mathcal{L} &= \mathbb{E}_{Q_x,Q_R}[\ln p(t_k, x_k, R_{\zeta\zeta_k})] \\
				&- \mathbb{E}_{Q_x}[\ln Q_x(x_k)] - \mathbb{E}_{Q_R}[\ln Q_R(R_{\zeta\zeta_k})].
			\end{aligned}
		\end{equation}
		By utilizing the augmented linear regression model \eqref{linreg}, the prior information $p(x_k)$ is intrinsically embedded into the pseudo-measurement vector $t_k$ and the augmented matrix $W_k$. Therefore, the complete data likelihood is entirely characterized by $p(t_k \mid x_k, R_{\zeta\zeta_k})$. Isolating the terms dependent on $x_k$, we define the state-dependent objective $\mathcal{L}_x$:
		\begin{equation}
			\mathcal{L}_x = \mathbb{E}_{Q_x}\Big[ \mathbb{E}_{Q_R}[\ln p(t_k \mid x_k, R_{\zeta\zeta_k})] \Big] - \mathbb{E}_{Q_x}[\ln Q_x(x_k)] + c_1.
		\end{equation}
		For the linear Gaussian observation model, the expected log-likelihood with respect to $Q_R$ has
		\begin{equation}
			\tiny
			\begin{aligned}
				\mathbb{E}_{Q_R}[\ln p(t_k \mid x_k, R_{\zeta\zeta_k})] &= -\frac{1}{2} (t_k - W_k x_k)^T \tilde{R}_{\zeta\zeta_k}^{-1} (t_k - W_k x_k)\\
				& + c_2.
			\end{aligned}
		\end{equation}
		Substituting this expected log-likelihood into $\mathcal{L}_x$ yields:
		\begin{equation}
			\begin{aligned}
				\mathcal{L}_x &= \int Q_x(x_k) \left[ -\frac{1}{2}\|t_k - W_k x_k\|^2_{\tilde{R}_{\zeta\zeta_k}^{-1}}\right] dx_k \\
				&- \int Q_x(x_k) \ln Q_x(x_k) dx_k + c_3.
			\end{aligned}
		\end{equation}
		To deduce the optimal distribution $Q_x^*(x_k)$, we expand the quadratic form inside the first integral with respect to $x_k$:
		\begin{equation}
			-\frac{1}{2} \Big( x_k^T (W_k^T \tilde{R}_{\zeta\zeta_k}^{-1} W_k) x_k - 2 x_k^T (W_k^T \tilde{R}_{\zeta\zeta_k}^{-1} t_k) \Big) + \text{const}.
		\end{equation}
		By completing the square, this expression is algebraically equivalent to the logarithm of an unnormalized Gaussian density $\ln \mathcal{N}(x_k \mid \mu^*, \Sigma^*)$, characterized by the information matrix and information vector:
		\begin{align}
			(\Sigma^*)^{-1} &= W_k^T \tilde{R}_{\zeta\zeta_k}^{-1} W_k, \label{eq:sigma_inv} \\
			(\Sigma^*)^{-1} \mu^* &= W_k^T \tilde{R}_{\zeta\zeta_k}^{-1} t_k. \label{eq:mu_inv}
		\end{align}
		Consequently, $\mathcal{L}_x$ can be elegantly formulated as the negative KL divergence between $Q_x(x_k)$ and this target Gaussian:
		\begin{equation}
			\mathcal{L}_x = -\mathrm{KL}\Big( Q_x(x_k) \,\big\|\, \mathcal{N}(x_k \mid \mu^*, \Sigma^*) \Big) + \text{c}.
		\end{equation}
		Because the KL divergence is strictly non-negative, $\mathcal{L}_x$ is maximized if and only if:
		\begin{equation}
			Q_x^*(x_k) = \mathcal{N}(x_k \mid \mu^*, \Sigma^*).
		\end{equation}
		Finally, we map the augmented matrices back to the standard Kalman state space to establish equivalence with Algorithm \ref{AlgSTKF}. Let the expected augmented precision be defined by the weight matrix $\tilde{R}_{\zeta\zeta_k}^{-1} = \operatorname{blkdiag}(M_p, M_r)$. Recalling the definitions $t_k = B_k^{-1}\binom{x_k^{-}}{y_k}$ and $W_k = B_k^{-1}\binom{I}{C}$, we substitute these into \eqref{eq:sigma_inv}:
		\begin{equation}
			(\Sigma^*)^{-1} = \begin{bmatrix} I \\ C \end{bmatrix}^T B_k^{-T} \begin{bmatrix} M_p & 0 \\ 0 & M_r \end{bmatrix} B_k^{-1} \begin{bmatrix} I \\ C \end{bmatrix}.
		\end{equation}
		Since $B_k = \operatorname{blkdiag}(B_p, B_r)$, the internal precision mapping yields the robustified prior and measurement covariances: $(B_p M_p^{-1} B_p^T)^{-1} = (\tilde{P}_k^-)^{-1}$ and $(B_r M_r^{-1} B_r^T)^{-1} = \tilde{R}_k^{-1}$. Expanding the blocks yields:
		\begin{equation}
			(\Sigma^*)^{-1} = (\tilde{P}_k^-)^{-1} + C^T \tilde{R}_k^{-1} C.
		\end{equation}
		Applying the Woodbury matrix identity, we obtain the variational covariance update:
		\begin{align}
			\Sigma^* &= \Big( (\tilde{P}_k^-)^{-1} + C^T \tilde{R}_k^{-1} C \Big)^{-1} \nonumber \\
			&= \tilde{P}_k^- - \tilde{P}_k^- C^T \left( C \tilde{P}_k^- C^T + \tilde{R}_k \right)^{-1} C \tilde{P}_k^-.
		\end{align}
		Defining the modified Kalman gain as 
		$$\tilde{K}_k = \tilde{P}_k^- C^T \left( C \tilde{P}_k^- C^T + \tilde{R}_k \right)^{-1},$$ 
		the covariance simplifies to:
		\begin{equation}
			\Sigma^* = (I - \tilde{K}_k C) \tilde{P}_k^-.
		\end{equation}
		The equivalent Joseph form is 
		\begin{equation}
			\Sigma^* = ({I}-\tilde{{K}}_k {C}){{P}}_{ k}^{-}({I}-\tilde{{K}}_k {C})^{T}+\tilde{{K}}_k {R}_{k}\tilde{{K}}_k^{T}.
		\end{equation}
		Expanding \eqref{eq:mu_inv} and left-multiplying by $\Sigma^*$ yields:
		\begin{equation}
			\mu^* = \Sigma^* (\tilde{P}_k^-)^{-1} x_k^- + \Sigma^* C^T \tilde{R}_k^{-1} y_k. \label{eq:mu_expand}
		\end{equation}
		Utilizing the identities $\Sigma^* (\tilde{P}_k^-)^{-1} = I - \tilde{K}_k C$ and $\Sigma^* C^T \tilde{R}_k^{-1} = \tilde{K}_k$, \eqref{eq:mu_expand} reduces to:
		\begin{align}
			\mu^* &= (I - \tilde{K}_k C) x_k^- + \tilde{K}_k y_k \nonumber \\
			&= x_k^- + \tilde{K}_k (y_k - C x_k^-).
		\end{align}
		The derived mean $\mu^*$ perfectly matches the MAP state update equation formulated in Algorithm \ref{AlgSTKF}. Thus, the STKF state update and Variational Bayesian inference are mathematically identical.
	\end{pf}
	
	\begin{remark}
		A byproduct of Theorem \ref{theoremiden} is that, instead of solving Problem \ref{prob2} using a set of coupled equations as shown in \cite{b6}, we solve the problem subsequently: (1) estimate the state $x_k$ using a fixed point iteration; (2) update the loss hyper-parameters and posterior error covariance (see Algorithm \ref{AlgSTKF-adap}). The two paths demonstrates highly consistent estimation accuracy (see the simulations), but our method possesses some advantages: On the one hand, it decouples the ``robustness process" and ``adaptive process", allowing us to design switching rules for complex noise scenarios. On the other hand, the decoupled method reduces algorithm complexity by calculating the posterior error covariance and update the hyper-parameter only once, rather than $N$ iterations as shown in \cite{b6}. 
	\end{remark}
	\subsection{Robust Adaptive Filtering}
	\label{robust-variational}
	According to Property \ref{p1} and Theorem \ref{theoremiden}, we have insights that optimizing \eqref{rloss1} actually jointly optimizes the state and covariance, where the diagonal element of the variance follows an inverse-Gamma distribution, which can be obtained by a fixed-point iteration as shown in Line 8-17 in Algorithm \ref{AlgSTKF}. We then focus on the update of hyper-parameters $\nu_i$ and $\tau_i^2$. Based on the variational inference~\cite{gelman1995bayesian,b6}, when applying \eqref{plambda} as a conjugate prior distribution, the  posterior hyper-parameters should be updated as
	\begin{equation}
		\begin{aligned}
			\nu_{i,k}^{+} & = \nu_{i,k}^{-} +1,\\
			\tau_{i,k}^2  & = (\textcolor{black}{\frac{\nu_{i,k}^{-}}{\nu_{i,k}^{+}}})(\tau_{i,k}^2)^{-} + \frac{e_{i,k}^{2}}{\nu_{i,k}^{+}} + \frac{[W_kP_kW_k^{T}]_{ii}}{\nu_{i,k}^{+}},
			\label{up1}
		\end{aligned}
	\end{equation}
	where $P_k$ is the posterior error covariance and $[*]_{ii}$ denotes the $i$-th diagonal element of matrix $*$. The proof of \eqref{up1} is available in Appendix \ref{vbderivation}. One observes that the update of $\nu_i$ is not related with the error $e_{i,k}$. In the iterative filtering, it is natural to assume that the hyper-parameter at the next time step is identical to the previous ones. However, this would induce unbounded $\nu_{i,k}^{+}$ and $\tau_{i,k}^2$. To alleviate this problem, we employ a forgetting factor $\rho_i \in (0,1)$ for channel $i$ in the hyper-parameter prediction step, which is motivated by \cite{b6} and given as follows: 
	\begin{equation}
		\begin{aligned}
			\nu_{i,k}^{-} &= \rho_i \nu_{i,k-1}^{+},\\
			\nu_{i,k}^{+} &= \nu_{i,k}^{-}+1.
		\end{aligned}
		\label{up2}
	\end{equation}
	
	In the extreme case $\rho=1$, the hyper-parameter update should be halted, which gives the following update:
	\begin{equation}
		\begin{aligned}
			\nu_{i,k}^{+}=\nu_{i,k-1}^{+},\quad (\tau_{i,k}^2)^{+} &=(\tau_{i,k-1}^2)^{-}= (\tau_{i,k-1}^2)^{+}.
		\end{aligned}
		\label{up3}
	\end{equation}	
	We name the new algorithm STKF-A and provide a summary of it in Algorithm \ref{AlgSTKF-adap}.
	\begin{prop}
		\label{propnew}
		The STKF-A is equivalent to the STKF when $\rho_i = 1$ for all $i = 1, 2, \ldots, l$. Additionally, it recovers the KF by setting $\rho_i = 1$ for all $i = 1, 2, \ldots, l$ and $\nu_i \to \infty$ for all $i$. 
	\end{prop}
	
	\begin{algorithm}[bt]
		\setstretch{1.0} 
		\caption{STKF-A}
		\label{AlgSTKF-adap}
		\begin{algorithmic}[1]
			\State \textbf{Initialization}\\
			Choose $\nu_i$ and $\tau_i^2$ for channel $i$, maximum iteration number $m_\mathrm{iter}$, and a threshold $\varepsilon$.
			\State \textbf{State Prediction}\\
			Run Line 4 to 7 in Algorithm \ref{AlgSTKF}.
			\State \textbf{Hyper-parameter Prediction}
			\For { $i=1$ to $l$}
			\If {$0<\rho_i<1$}
			\State Execute \eqref{up2} to obtain \textcolor{black}{$\nu_{i,k}^{+}$}  
			\ElsIf{$\rho_i=1$} 
			\State Execute \eqref{up3} to obtain \textcolor{black}{$\nu_{i,k}^{+}$} 
			\EndIf
			\EndFor
			\State \textbf{State Update}\\
			Run Line 8-17 in Algorithm \ref{AlgSTKF} to obtain $\hat{x}_k$ and $P_k$
			\State \textbf{Hyper-parameter Update}
			\For {$i=1$ to $l$}
			\If { $0<\rho_i<1$}
			\State Execute \eqref{up1} to update $(\tau_{i,k}^{2})^{+}$
			\ElsIf { $\rho_i=1$}
			\State Execute \eqref{up3} to update $(\tau_{i,k}^{2})^{+}$
			\EndIf
			\EndFor
		\end{algorithmic}
	\end{algorithm}  
	
	This can be proved by comparing different estimators by substituting the values of $\rho_i$ and $\nu_i$. We then consider the case that the noise variance is time-varying and is occasionally polluted by outliers. To avoid the harmful effects of outliers on the adaptive mechanism, we introduce a probabilistic switching rule based on a Bernoulli distribution to infer the probability of a measurement being an outlier.
	
	To mitigate the impact of outliers on the adaptive mechanism, we introduce a Bernoulli indicator $Z_{i,k} \in \{0, 1\}$ for the $i$-th channel, with prior outlier probability $P(Z_{i,k}=1)=\pi$. The likelihoods of the innovation $e_{i,k}$ under normal and outlier states are modeled as $L_0=\mathcal{N}(e_{i,k} \mid 0, S_{i,k}^{(0)})$ and $L_1=\mathcal{N}(e_{i,k} \mid 0, S_{i,k}^{(1)})$ where $S_{i,k}^{(0)}=(\tau_{i,k}^2)^{-}$ and $S_{i,k}^{(1)}=9(\tau_{i,k}^2)^{-}$. By Bayes' theorem, the posterior probability of an outlier is:$$\gamma_{i,k}=\frac{\pi L_1}{\pi L_1+(1-\pi)L_0}.$$
	The covariance update is given as
	\begin{equation}
		\begin{aligned}
		(\tau_{i,k}^2)^{+} &= (1-\gamma_{i,k}) \left( \frac{\nu_{i,k}^{-}}{\nu_{i,k}^{+}}(\tau_{i,k}^2)^{-} 
		+ \frac{e_{i,k}^{2} + [W_k P_k W_k^{T}]_{ii}}{\nu_{i,k}^{+}} \right) \\
		&+ \gamma_{i,k} (\tau_{i,k}^2)^{-}
		\label{up_soft}.
		\end{aligned}
	\end{equation}
	This continuous transition strategy, designated as STKF-AR, is summarized in Algorithm \ref{AlgSTKF-adap-robust}.

	\begin{algorithm}[bt]
		\setstretch{1.0} 
		\caption{STKF-AR}
		\label{AlgSTKF-adap-robust}
		\begin{algorithmic}[1]
			\State {Run Line 1 to 14 in Algorithm \ref{AlgSTKF-adap}}
			\State \textbf{Hyper-parameter Update}
			\For {$i=1$ to $l$}
			\If { $0<\rho_i<1$}
			\State Execute \eqref{up_soft} to obtain $(\tau_{i,k}^{2})^{+}$
			\ElsIf { $\rho_i=1$}
			\State Execute \eqref{up3} to update $(\tau_{i,k}^{2})^{+}$
			\EndIf
			\EndFor
		\end{algorithmic}
	\end{algorithm}  
	\subsection {Convergence Rate Analysis of the Unknown Covariance}
	\label{convrate}
	We focus on the steady-state behavior of $(\tau_{i,k}^{2})^{+}$. Based on \eqref{up2}, in steady state, one has
	\[
	\nu_{i,\infty}^{+} = \nu_{i,\infty}^{-}+1=\rho_i \nu_{i,\infty}^{+}+1.
	\]
	It follows that 
	\begin{equation}
		\nu_{i,\infty}^{+}=\frac{1}{1-\rho_i}, \quad \nu_{i,\infty}^{-}=\frac{\rho_i}{1-\rho_i}.
		\label{steadynu}
	\end{equation} 
	Denoting $\lambda_{i,k}^{+} \triangleq (\tau_{i,k}^{2})^{+}$ and substituting \eqref{steadynu} into \eqref{up1} yields
	\begin{equation}
		\lambda_{i,k}^{+}= \rho_i \lambda_{i,k-1}^{+} + ({1-\rho_i}) e_{i,k}^{2} + (1-\rho_i)w_{i,k}P_k w_{i,k}^{T},
		\label{propogate}
	\end{equation} 
	where $e_{i,k}=t_{i,k}-w_{i,k}\hat{x}_k$ is the measurement residual, and $w_{i,k}$ is the $i$-th row of $W_k$. Defining the realized variational error as $\mathcal{E}_{i,k} \triangleq e_{i,k}^{2}+ w_{i,k}P_k w_{i,k}^{T}$, we establish the following theorem.
	
	\begin{theorem}
		\label{theorem5}
		Suppose the forgetting factor $\rho_i \in (0,1)$ and the realized variational error sequence $\mathcal{E}_{i,k}$ is weakly stationary and uncorrelated, with mean $\mathbb{E}[\mathcal{E}_{i,k}] = \sigma_i^2$ and variance $\mathrm{Var}(\mathcal{E}_{i,k}) = \eta_i^4$. Then, the scale parameter sequence $\lambda_{i,k}^{+}$ defined by the recursive update \eqref{propogate} converges in mean and variance as $k \to \infty$. Specifically, the steady-state mean and variance are given by:
		\begin{align}
			\lim_{k \to \infty} \mathbb{E}[\lambda_{i,k}^{+}] &= \sigma_i^2, \label{thm_mean} \\
			\lim_{k \to \infty} \mathrm{Var}(\lambda_{i,k}^{+}) &= \frac{1-\rho_i}{1+\rho_i} \eta_i^4. \label{thm_var}
		\end{align}
	\end{theorem}
	The proof is available at \ref{prooftheorem5}. We then analyze the transient behavior of $\operatorname{E}(\lambda_{i,k}^{+})$. 
	
	Taking the expectation of \eqref{propogate} and substituting the stationary mean $\mathbb{E}[\mathcal{E}_{i,k}] = \sigma_i^2$ yields the recursive relation:
	\begin{equation}
		\mathbb{E}[\lambda_{i,k}^{+}] = \rho_i \mathbb{E}[\lambda_{i,k-1}^{+}] + (1-\rho_i)\sigma_i^2.
		\label{mean_recursion}
	\end{equation}
	Subtracting $\sigma_i^2$ from both sides and recursively unrolling the sequence to the initial step $k=0$ provides the explicit expression for the transient mean:
	\begin{equation}
		\mathbb{E}[\lambda_{i,k}^{+}] = \sigma_i^2 + \rho_i^k \left( \mathbb{E}[\lambda_{i,0}^{+}] - \sigma_i^2 \right).
		\label{transient_mean}
	\end{equation}
	Equation \eqref{transient_mean} characterizes the estimator's transient behavior. The term $\rho_i^k \left( \mathbb{E}[\lambda_{i,0}^{+}] - \sigma_i^2 \right)$ represents the transient bias, which decays exponentially to zero. The convergence rate is governed by the forgetting factor $\rho_i$, defining a tracking time constant 
	\begin{equation}
		\tau_i = -1/\ln(\rho_i) \approx 1/(1-\rho_i).
		\label{covnspeed}
	\end{equation} 
	\begin{remark}
		A smaller $\rho_i$ accelerates the decay of the transient bias, leading to faster convergence to the steady-state mean $\sigma_i^2$. However, recalling the steady-state variance in \eqref{thm_var}, a fundamental trade-off emerges: decreasing $\rho_i$ improves transient tracking speed but amplifies the steady-state variance $\frac{1-\rho_i}{1+\rho_i}\eta_i^4$. Conversely, a $\rho_i$ closer to $1$ yields a smoother steady-state estimate with lower variance, at the cost of a prolonged transient phase. Therefore, the selection of $\rho_i$ need balance the convergence rate against steady-state estimation precision.
	\end{remark}
	\subsection{Parameter Selection and Discussion}
	Our proposed framework relies on three channel-specific hyperparameters: the scale parameter $\tau_i^2$, the DOF $\nu_i$, and the decay coefficient $\rho_i$. Because the normalized nominal covariance is $E(\zeta_k\zeta_k^T)=I$, we set $\tau_i^2=1$. The parameter $\nu_i$ reflects confidence in the nominal covariance; smaller values increase the temporary variance \eqref{inflationeffect} to reject outliers, at the cost of degraded performance under pure Gaussian noise. Consequently, we recommend $\nu_i \to \infty$ for Gaussian channels and $\nu_i \in [0.5, 5]$ for outlier-prone channels.
	
	In the adaptive filters, the decay coefficient $\rho_i$ is crucial to prevent $\nu_{i,k}$ from growing unbounded, which would otherwise erroneously recover the standard KF. As established in Theorems 5, $\rho_i$ balances tracking speed against variance smoothness. We advise setting $\rho_i \in [0.95, 1)$. Setting $\rho_i=1$ recovers the STKF, which further simplifies to the standard KF if $\nu_i \to \infty$.
	
	Robustness and adaptability are competing objectives: robustness temporarily inflates the latent variance while keeping the prior fixed, whereas adaptability continuously updates the prior \eqref{up1}. To unify them, a probabilistic switching rule with Bernoulli indicator is utilized. An advantage of this framework is its channel-level flexibility. It allows practitioners to encode prior knowledge directly into the algorithm—for example, assigning static robust parameters (small $\nu_i$, $\rho_i=1$) to channels with frequent outliers, while enabling dynamic adaptation ($\rho_i < 1$) for channels experiencing time-varying noise.
	
	\section{Simulations}
	This section evaluates the proposed algorithms through four numerical examples.
	\subsection{Example 1}
	\label{identitysim}
	We consider the following tracking problem:
	\begin{equation}
		\begin{aligned}
			x_{k+1}&=A x_k + Bw_k \\
			y_k &= C x_k +v_k
		\end{aligned}
		\label{TargetTracking}
	\end{equation}
	where $A=\begin{bmatrix}1&T\\
		0&1\end{bmatrix}$, $B=\begin{bmatrix}0.5T^2\\
		T\end{bmatrix}$, $w_k\sim \mathcal{N}(0,1)$, $C=\begin{bmatrix}1&0\end{bmatrix}$, and $T=0.01$ is the sampling time. We consider the following two types of measurement noise:
	\[
	\begin{cases}
		\text{Case 1:\quad} v_k \sim 0.95\mathcal{N}(0,0.1)+0.05\mathcal{N}(0,10)\\
		\text{Case 2:\quad} v_k \sim \mathcal{N}(0,R_{t}),
	\end{cases}
	\]
	where $R_{t}= [2(\sin(0.04\pi k T))^{2}+1]R$ and $R=0.1$ in Case 2. 
	
	In Case 1, we conduct simulations investigating the effects of $\nu_3$ in STKF and the result is shown in Fig. \ref{rmsenu}. One can see that the error performance of the STKF first decays with the increment of $\nu_3$ and then increases slowly, and finally coincides with the KF. In robust filtering scenarios, a too small $\nu_i$ may induce instability and a too large $\nu_i$ would damage the robustify effect. 
	
	\begin{figure}[htp]
		\centering
		\includegraphics[width=0.65\columnwidth]{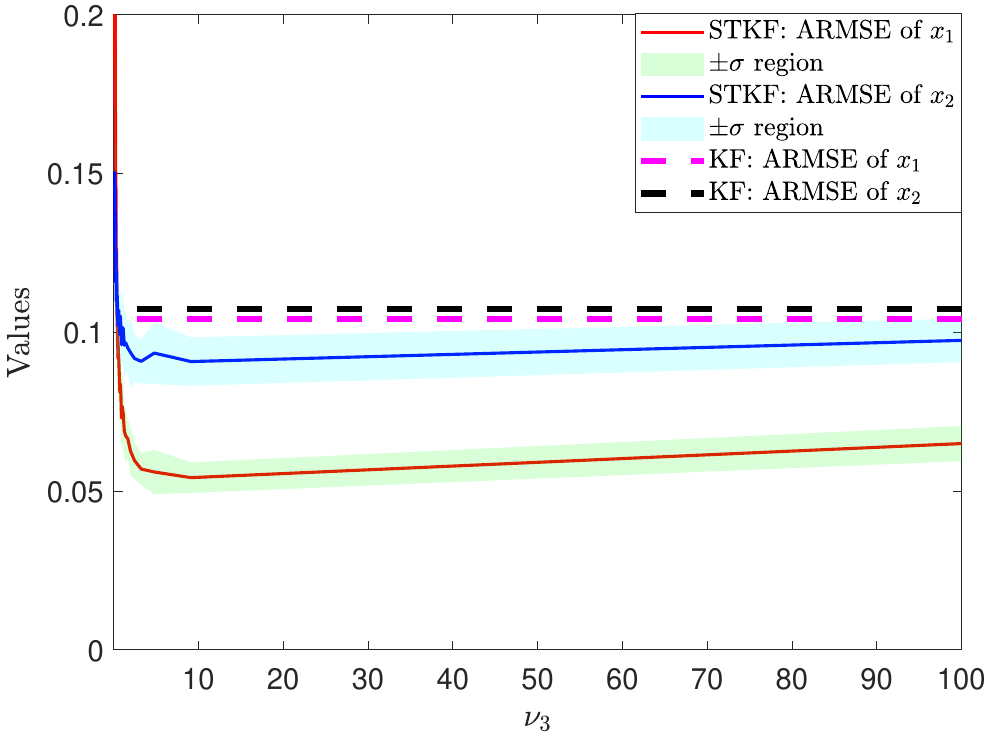}
		\caption{Average RMSE (ARMSE) with different $\nu_3$ in STKF.}
		\label{rmsenu}
	\end{figure}
	
	In Case 2 with adaptive measurement noise, we set the initial process and measurement covariance as $Q=BB^{T}$, $R=0.1$, and use $\rho=0.99$ in VBKF. As in STKF-A, we apply the same initial process and measurement covariance as is used in VBKF. Moreover, we set $\rho_1=\rho_2=1$, $\rho_3=0.99$, $\tau_i^2=1$ for $i=1,2,3$, $\nu_1=\nu_2=10^{8}$, and $\nu_3=100$. The estimated covariance obtained by STKF-A and by VBKF, as well as the ground measurement covariance, are visualized in Fig. \ref{adapR}. The corresponding error performances are summarized in Table \ref{identTable}. We find that the performance of STKF-A is identical to VBKF, but with a smaller average iteration number, benefiting from the break condition as shown in Line 10 of Algorithm \ref{AlgSTKF}. Both STKF-A and VBKF are significantly better than KF. 
	\begin{figure}[!htp]
		\centering
		\includegraphics[width=0.6\columnwidth]{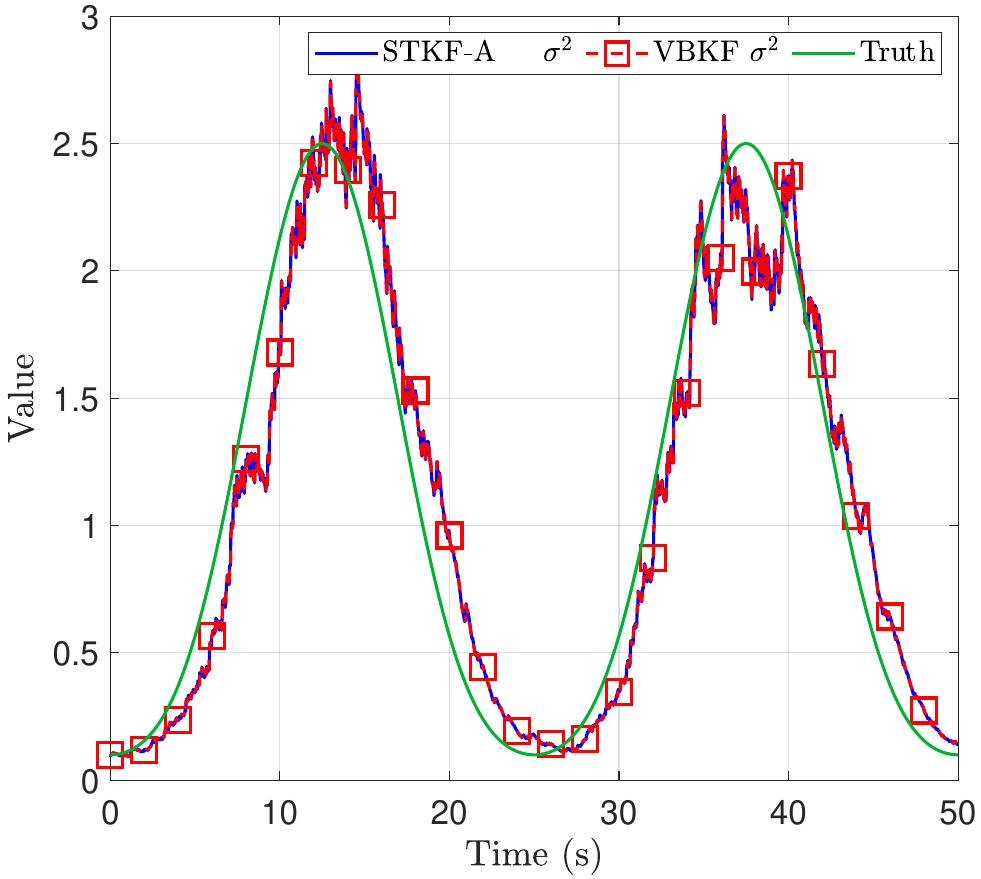}
		\caption{The measurement noise covariance (or variance) tracking performance of VBKF and STKF-A.}
		\label{adapR}
	\end{figure}
	\begin{table}[]
		\centering
		\caption{Performance comparison of different estimators in different cases.}
		\scalebox{0.7}{
			\begin{tabular}{c|cccc}
				\hline
				\hline
				\multirow{2}{*}{Scenario} & \multicolumn{1}{r}{\multirow{2}{*}{Estimators}} & \multirow{2}{*}{\begin{tabular}[c]{@{}c@{}}RMSE of \\ $x_1$\end{tabular}} & \multirow{2}{*}{\begin{tabular}[c]{@{}c@{}}RMSE of \\ $x_2$\end{tabular}} & \multirow{2}{*}{\begin{tabular}[c]{@{}c@{}} avg. iteration \\ number\end{tabular}} \\
				
				& \multicolumn{1}{r}{}   
				&                                                                         &                                                                         &                                                                              \\		
				\hline
				\multirow{3}{*}{Case 1}   & \multicolumn{1}{c}{VBKF-fixed}                  &     0.053                                                                    &    0.081                                                                     &     4                                                                         \\
				& \multicolumn{1}{c}{STKF}                        &     0.053                                                                    &    0.081                                                                     &    1.096                                                                         \\
				& {KF}                          &   0.111                                                                      &  0.129                                                                      &      1.0                                                                       \\
				\hline
				\multirow{3}{*}{Case 2}   & VBKF                                            &    0.092                                                 &    0.086                                            & 4                                                                            \\
				& STKF-A                                     &   0.092                                                &    0.086                                                &      2.190                                                                        \\
				& KF                                              &     0.130                                               &   0.143                                                &   1.0   \\                                                    
				\hline
				\hline                   
		\end{tabular}}
		\label{identTable}
	\end{table}
	\subsection{Example 2}
	Following system dynamics \eqref{TargetTracking}, we keep $\rho_1=\rho_2=1$ and investigate the effect of $\rho_3=\rho$  by considering the following step-like measurement covariance:
	\begin{equation}
		v_k \sim
		\begin{cases}
			\mathcal{N}(0,0.1), k \le 2000 \\
			\mathcal{N}(0,2.5), 2000 < k \le 4000 \\
			\mathcal{N}(0,0.1), k \ge 4000.
		\end{cases}
	\end{equation}
	In the simulation, we compare the convergence speed of $\rho=0.995$, $\rho=0.99$, $\rho=0.98$, $\rho=0.97$ in Fig. \ref{rho}, where the blue line denotes the ground truth variance, and the dashed black line is obtained by Theorems \ref{theorem5} and \eqref{covnspeed}. We observe that practical convergence speed fits with the theoretical analysis very well, which verified the theorem. We further visualize the convergence speed and the steady state variance with respect to $\rho$ and give the corresponding RMSEs in Fig. \ref{converho}. The results highlight a trade-off between the convergence speed and the steady state variance. 
	\begin{figure}[htp]
		\centering
		\includegraphics[width=1.0\columnwidth]{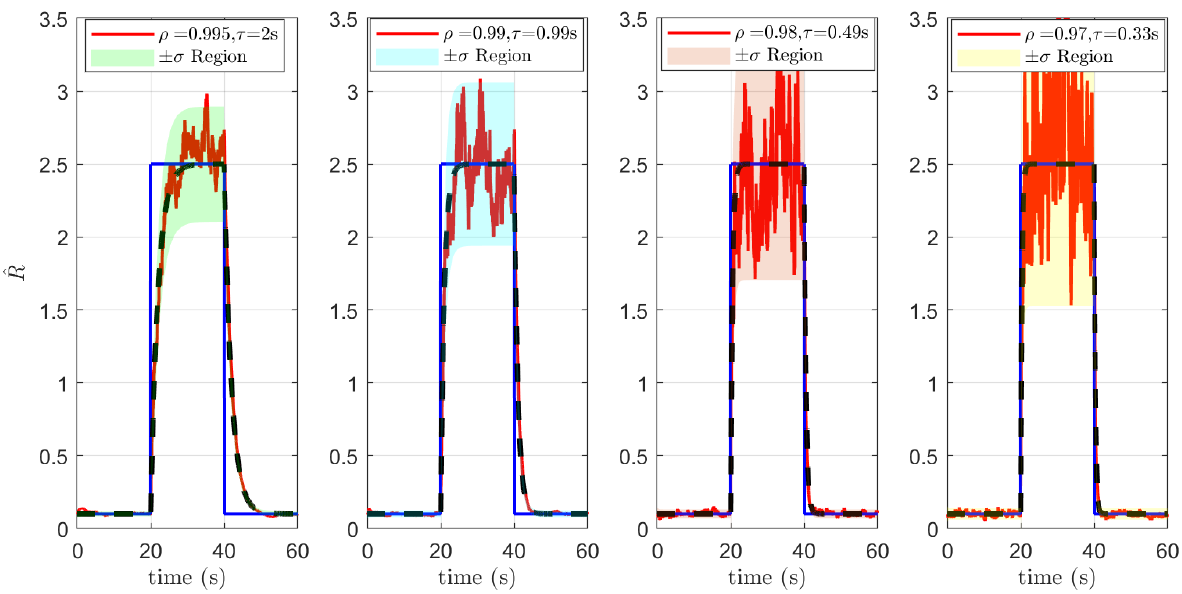}
		\caption{Theoretical (based on Theorem \ref{theorem5}) and practical variance convergence rate and the corresponding estimation variance with different $\rho$. The time constant is obtained by \eqref{covnspeed}, expressed in seconds.}
		\label{rho}
	\end{figure}
	\begin{figure}[!htp]
		\centering
		\subfigure[]{
			\begin{minipage}[t]{0.49\linewidth}
				\centering
				\includegraphics[width=1\columnwidth]{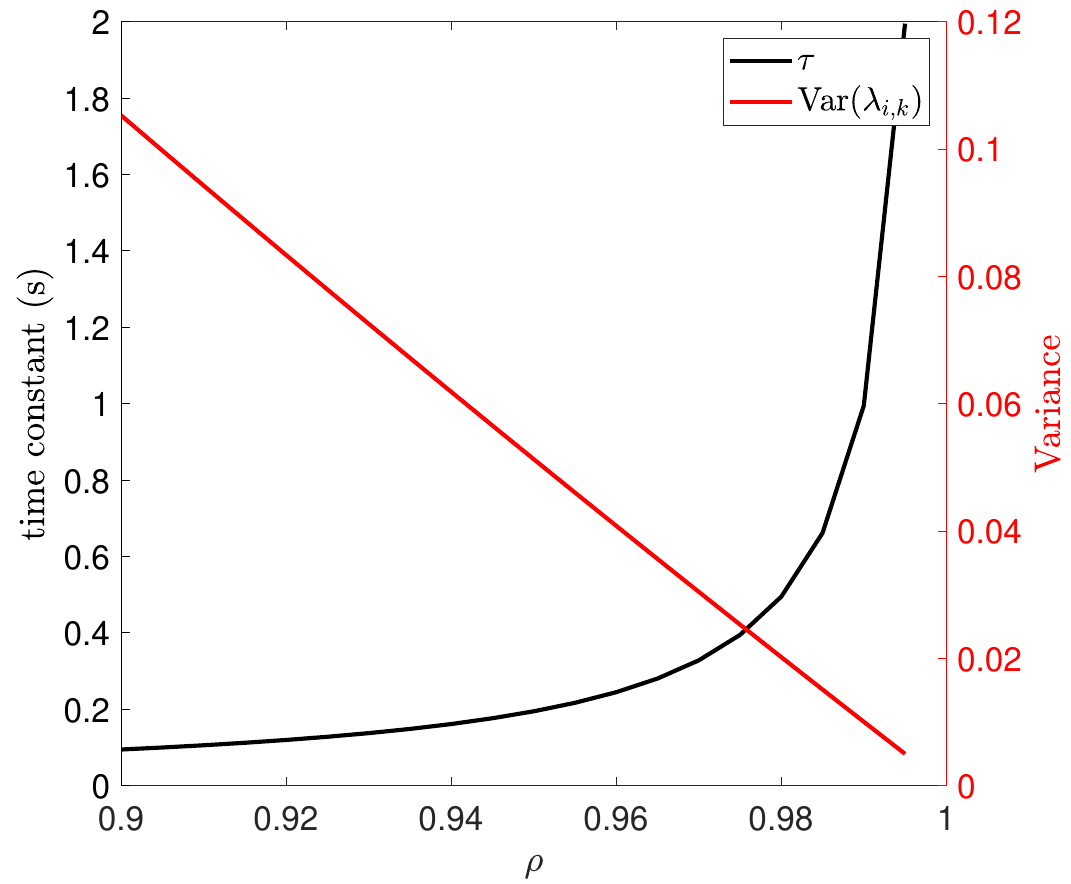}
				\label{TimeConstant}
			\end{minipage}%
		}%
		\subfigure[]{
			\begin{minipage}[t]{0.45\linewidth}
				\centering
				\includegraphics[width=1\columnwidth]{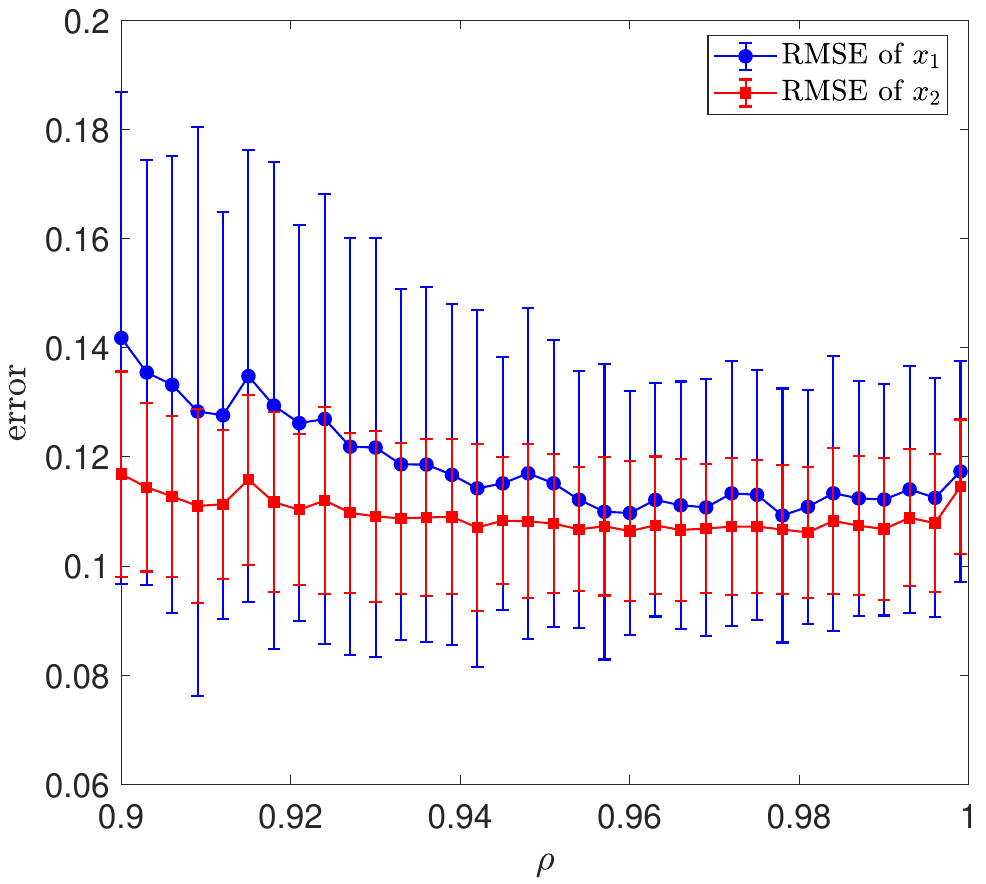}
				\label{RMSErho}
			\end{minipage}%
		}%
		\caption{The trade-off effects of $\rho$ in STKF-A. (a) The trade-off between convergence time constant and convergence variance regarding $\rho$. (b) The error performance with different $\rho$.}
		\label{converho}
	\end{figure}
	
	\subsection{Example 3}
	\label{super}
	We consider a 1-DOF torsion load system with unknown disturbances as given in \cite{zhao2016fusion,luan2025kalman}. The discrete system dynamics, with sampling time of $dt=0.01$ and maximum time step $N_{t}=2000$, are given by
	\begin{equation}
		\begin{aligned}
			{x}_{k}&={F}_k {x}_{k-1}+{G}_{1,k} {u}_{k-1}+{G}_{2,k} {d}_{k-1}+{w}_{k}\\
			{y}_k&={H}_k {x}_{k}+{v}_{k}
			\label{linearfun}
		\end{aligned}
	\end{equation}
	with
	\begin{equation}\nonumber
		\tiny
		\begin{aligned}
			F_k&=\begin{bmatrix}
				0.9205 &   0.0795 &   0.0085 &   0.0003\\
				0.2045 &   0.7955 &   0.0007 &   0.0085\\
				-14.3468 &  14.3468 &   0.6872 &   0.0746\\
				37.5370 & -37.5370 &   0.1863 &   0.6405
			\end{bmatrix}\\
			G_{1,k}&=\begin{bmatrix}
				0.0826,
				0.0031,
				15.5568,
				1.2100
			\end{bmatrix}^{T}\\
			G_{2,k}&=\begin{bmatrix}
				0.0031,
				0.2076,
				1.2100,
				38.7470
			\end{bmatrix}^{T} \\
			H_k&=\begin{bmatrix}
				1&0&0&0\\
				0&1&0&0
			\end{bmatrix}.
		\end{aligned}
	\end{equation}
	where $x_k=[\theta_m,\theta_t,v_m,v_t]^{T} \in \mathbb{R}^{4}$ represents the state vector, consisting of the angles at both the motor and load sides, as well as the velocities at the motor and load sides, $u_k \in \mathbb{R}$ is the motor torque, $d_k \in \mathbb{R}$ is the unknown disturbance, $y_k$ is the noisy angle measurements at both motor and load sides. We  augment the disturbance and the state as a new state and use a random walk model for the unknown disturbance dynamics by analogy to \cite{li2023generalized} to simultaneously estimate the disturbance and state, which gives 
	\begin{equation}
		\begin{aligned}
			\mathrm{x}_k &= \mathrm{A}_k \mathrm{x}_{k-1} + \mathrm{B}_k \mathrm{u}_{k-1} + \mathrm{w}_k \\
			\mathrm{y}_k &= \mathrm{C}_k \mathrm{x}_k  + \mathrm{v}_k
		\end{aligned}
		\label{sys0}
	\end{equation}
	where 
	\begin{equation}\nonumber
		\begin{aligned}
			\mathrm{x}_{k}&=\begin{bmatrix}
				d_{k} \\
				x_{k}
			\end{bmatrix}, \quad \mathrm{w}_{k}=\begin{bmatrix}
				w_{d,k} \\
				w_{x,k}
			\end{bmatrix} ,\quad \mathrm{A}_k=\begin{bmatrix}
				1   &0_{1 \times 4}\\
				G_{2,k} &F_k 
			\end{bmatrix},\\
			\mathrm{B}_k&=\begin{bmatrix}
				0\\
				G_{1,k}
			\end{bmatrix},\quad\mathrm{C}_k =\begin{bmatrix}
				{0}_{2 \times 1} & H_k
			\end{bmatrix}, \quad\mathrm{y}_k={y}_k.
		\end{aligned}
		\label{notdef}
	\end{equation}
	
	In the simulation, we consider the following three complex noise scenarios: (1) the process noise covariance is step-like where as the measurement covariance keeps constant; (2) the process noise is contaminated by outliers, where as the measurement noise covariance changes continually; (3) the process noise covariance keep constant, whereas the measurement noise covariance changes continually and is occasionally polluted by outliers, i.e.,
	\begin{equation}
		\tiny
		\begin{aligned}
			&\text{Case 1} \begin{cases}
				\mathrm{w}_k \sim \mathcal{N}(0,Q), 800\le k< 1200\\
				\mathrm{w}_k \sim \mathcal{N}(0,100Q), \text{otherwise}\\
				\mathrm{v}_k \sim \mathcal{N}(0,R), 
			\end{cases}
			\\
			&\text{Case 2}
			\begin{cases}
				\mathrm{w}_k \sim 0.99\mathcal{N}(0,Q)+0.01\mathcal{N}(0,900Q)\\
				\mathrm{v}_k \sim \mathcal{N}(0,R_t)
			\end{cases}\\
			&\text{Case 3} 	\begin{cases}
				\mathrm{w}_k \sim \mathcal{N}(0,Q)\\
				\mathrm{v}_k \sim 0.99\mathcal{N}(0,R_t)+0.01\mathcal{N}(0,900R)
			\end{cases}
		\end{aligned}
	\end{equation}
	with
	\begin{equation}
		\begin{aligned}
			Q&=\begin{bmatrix}
				0.01&0\\
				0&0.01\mathrm{I}_4
			\end{bmatrix},  R=0.5\mathrm{I}_2, \\
			R_t& = [2(\sin(0.04\pi k T))^{2}+1]R.
		\end{aligned}
		\label{noisemath}
	\end{equation}
	We compare the performance of VBKF, STKF-A (or STKF-AR), KF-DOB, as well as RBKF1, RBKF2, and RBKF3 (by analogously applying Rows 2, 3, and 4 of Table \ref{costfun} in Algorithm \ref{AlgSTKF}). Specifically, in all estimators, we apply the same nominal process covariance using $Q^{*}=Q$ and $R^{*}=R$, where $Q$ and $R$ are shown in \eqref{noisemath}. The hyper-parameters of different estimators are tuned based on the properties of their losses and the characteristics of the noise, and are summarized in Table \ref{comTable}. Results demonstrate that the proposed methods are the best among the all filters.
	
	\begin{table*}[]
		\centering
		\caption{Performance comparison of different estimators in complex noise scenarios.}
		\scalebox{0.75}{
			\begin{tabular}{c|ccccccc}
				\hline
				\hline
				Scenario& Estimators & Hyper-parameters & RMSE of  $x_1$ & RMSE of  $x_2$& RMSE of  $x_3$& RMSE of  $x_4$& RMSE of  $x_5$ \\
				
				\hline
				\multirow{6}{*}{Case 1}   & {VBKF} &$\rho=0.98$             &0.710  &2.609  &2.717  &71.101  &111.818                                                                      \\
				&{STKF-A}     & \scalebox{0.85}{\begin{tabular}{c}
						$\mathbf{\nu_p}=100\cdot \mathbf{1}_5$,~$\mathbf{\nu_r}=10^{8}\mathbf{1}_2$\\
						$\mathbf{\rho_p}=0.98\cdot\mathbf{1}_5$,~$\mathbf{\rho_r}=\mathbf{1}_2$\\
				\end{tabular}}              &0.507  &0.237  &0.249  &58.162  &64.819                                                                       \\
				& {KF}       & ${\emptyset}$                   &0.604  &0.589  &0.477  &55.350  &71.400                                                                   \\
				& {RBKF1}  &\scalebox{0.85}{ \begin{tabular}{c}
						$\mathbf{\nu_p}=2\cdot \mathbf{1}_5$,~
						$\mathbf{\nu_r}=10^{8}\cdot \mathbf{1}_2$\\
				\end{tabular}}
				&0.557  &0.258  &0.265  &55.888  &62.650  \\
				& {RBKF2}       &   \scalebox{0.85}{ \begin{tabular}{c}
						$\mathbf{\nu_p}=0.5\cdot \mathbf{1}_5$,~
						$\mathbf{\nu_r}=1.999\cdot \mathbf{1}_2$\\
				\end{tabular}}                &0.535  &0.254  &0.259  &55.664  &60.958   \\
				& {RBKF3}  & \scalebox{0.85}{ \begin{tabular}{c}
						$\mathbf{\nu_p}=1\cdot \mathbf{1}_5$,~
						$\mathbf{\nu_r}=10^{8}\cdot \mathbf{1}_2$\\
				\end{tabular}}       &0.555  &0.273  &0.271  &55.093  &60.377   \\
				\hline
				\multirow{6}{*}{Case 2}  &$\rho=0.98$ & VBKF  &0.408 &0.625 &0.863 &32.089 &58.726       \\
				& {STKF-AR}    &\scalebox{0.75}{\begin{tabular}{c}
						$\mathbf{\nu_p}=3\cdot \mathbf{1}_5$,~$\mathbf{\nu_r}=100\mathbf{1}_2$\\
						$\mathbf{\rho_p}=1\cdot\mathbf{1}_5$,~$\mathbf{\rho_r}=0.98\cdot\mathbf{1}_2$\\
				\end{tabular}}              &0.375 & 0.423 & 0.503 & 28.231 & 47.228                                                                      \\
				& KF & $\emptyset$  &0.468 & 0.416 & 0.515 & 28.315 & 49.509  \\         
				& {RBKF1}   &\scalebox{0.85}{ \begin{tabular}{c}
						$\mathbf{\nu_p}=2\cdot \mathbf{1}_5$,~
						$\mathbf{\nu_r}=2\cdot \mathbf{1}_2$\\
				\end{tabular}}  &0.392 & 2.151 & 2.263 & 35.268 & 62.601  \\
				& {RBKF2}   &\scalebox{0.85}{ \begin{tabular}{c}
						$\mathbf{\nu_p}=0.5\cdot \mathbf{1}_5$,~
						$\mathbf{\nu_r}=0.5\cdot \mathbf{1}_2$\\
				\end{tabular}}                     &0.380 & 0.743 & 0.978 & 32.881 & 59.835   \\
				& {RBKF3}  &\scalebox{0.85}{ \begin{tabular}{c}
						$\mathbf{\nu_p}= \mathbf{1}_5$,~
						$\mathbf{\nu_r}= \mathbf{1}_2$\\
				\end{tabular}}  &0.386 & 0.502 & 0.754 & 30.510 & 56.959   \\ 
				\hline
				\multirow{6}{*}{Case 3}   & VBKF  &$\rho=0.98$  &0.200 &0.359 &0.406 &15.085  &23.448          \\
				& {STKF-AR}&  \scalebox{0.85}{\begin{tabular}{c}
						$\mathbf{\nu_p}=10^{8}\cdot \mathbf{1}_5$,~$\mathbf{\nu_r}=100\mathbf{1}_2$\\
						$\mathbf{\rho_p}=\mathbf{1}_5$,~$\mathbf{\rho_r}=0.98 \cdot\mathbf{1}_2$\\
				\end{tabular}}  &0.203 &0.322 &0.389 &15.165  &22.948                                                                      \\
				& KF   & $\emptyset$      &0.359 &0.600 &0.524 & 21.320 & 34.781 \\         
				& {RBKF1}  &\scalebox{0.85}{ \begin{tabular}{c}
						$\mathbf{\nu_p}= 10^{8}\cdot\mathbf{1}_5$,~
						$\mathbf{\nu_r}= 2\cdot\mathbf{1}_2$\\
				\end{tabular}}  &0.229 &0.428 &0.508 &16.257  &24.640                                                                   \\
				& {RBKF2}  &\scalebox{0.85}{ \begin{tabular}{c}
						$\mathbf{\nu_p}= 1.999\cdot\mathbf{1}_5$,~
						$\mathbf{\nu_r}= 0.5 \cdot\mathbf{1}_2$\\
				\end{tabular}}   &0.210 &0.363 &0.425 & 15.576 &23.530    \\
				& {RBKF3}    &\scalebox{0.85}{ \begin{tabular}{c}
						$\mathbf{\nu_p}= 10^{8}\cdot\mathbf{1}_5$,~
						$\mathbf{\nu_r}= \mathbf{1}_2$\\
				\end{tabular}}   &0.215 &0.341 &0.411 &15.554  & 23.630  \\                                          
				\hline
				\hline                   
		\end{tabular}}
		\label{comTable}
	\end{table*}
	
	We further visualize the measurement covariance tracking of VBKF and STKF-AR in Case 2 and Case 3 in Figs. \ref{case2} and \ref{case3}, respectively. We find that in both cases, the conventional VBKF fails to track the ground truth noise covariance, where the tracking process is destroyed by the process outliers in Case 2 and the measurement outliers in Case 3. Our proposed STKF-AR mitigates these problems, since the proposed methods address both the outliers and adaptive noise in a unified framework. 
	
	\begin{figure}[!htp]
		\centering
		\subfigure[VBKF]{
			\begin{minipage}[t]{0.49\linewidth}
				\centering
				\includegraphics[width=1\columnwidth]{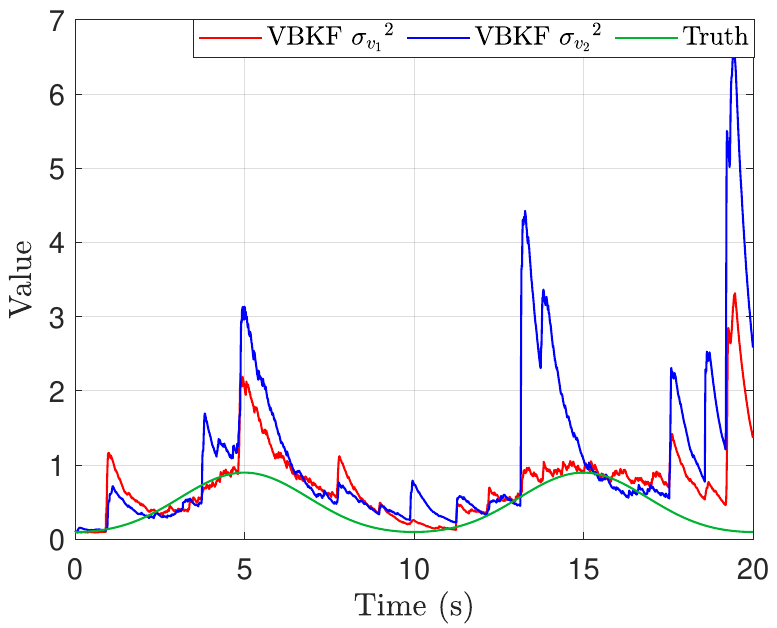}
				\label{case2vekf}
			\end{minipage}%
		}%
		\subfigure[STKF-AR]{
			\begin{minipage}[t]{0.49\linewidth}
				\centering
				\includegraphics[width=1\columnwidth]{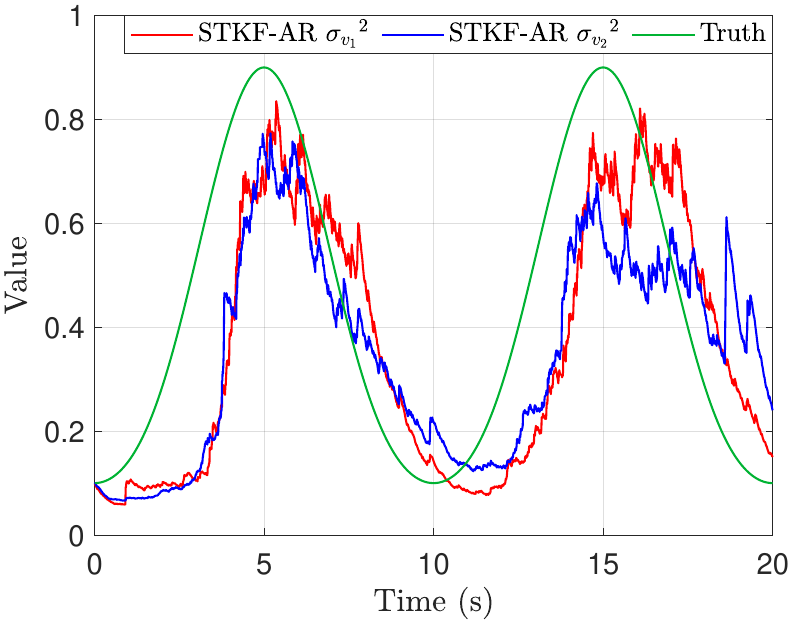}
				\label{case2stkf}
			\end{minipage}%
		}%
		\caption{The measurement covariance (or variance) tracking performance of VBKF and STKF-AR in Case 2. The blue and red lines denote the estimated variance, and the green line denote the ground truth variance. (a) The performance of VBKF. (b) The performance of STKF-AR.}
		\label{case2}
	\end{figure}
	
	\begin{figure}[!htp]
		\centering
		\subfigure[VBKF]{
			\begin{minipage}[t]{0.49\linewidth}
				\centering
				\includegraphics[width=1\columnwidth]{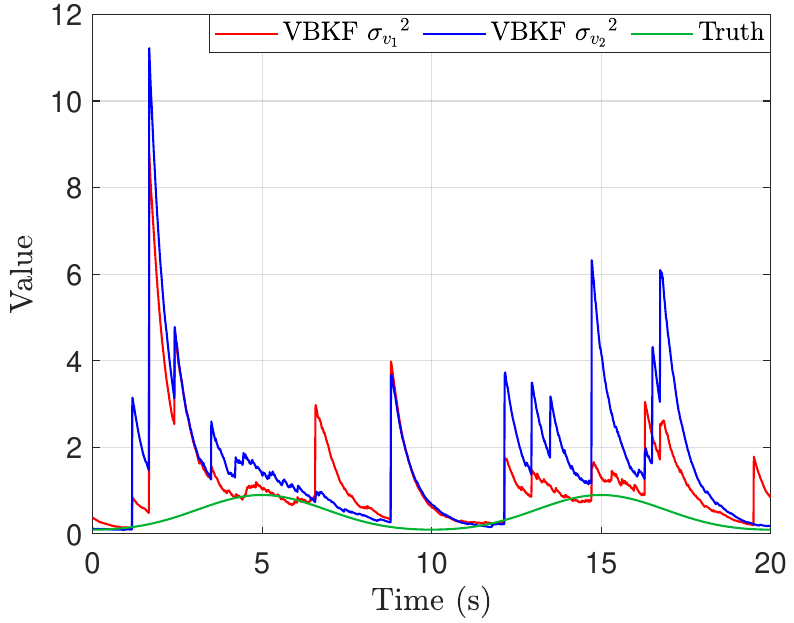}
				\label{case3vekf}
			\end{minipage}%
		}%
		\subfigure[STKF-AR]{
			\begin{minipage}[t]{0.49\linewidth}
				\centering
				\includegraphics[width=1\columnwidth]{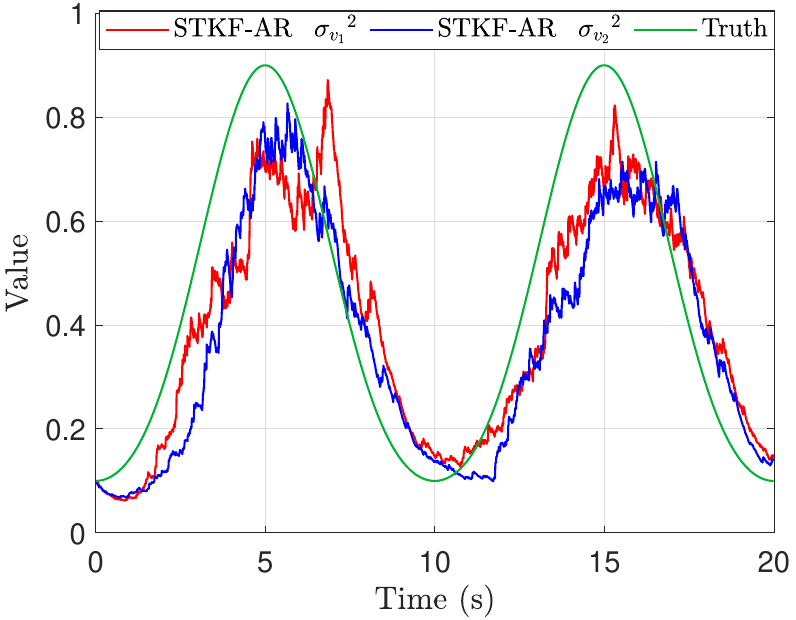}
				\label{case3stkf}
			\end{minipage}%
		}%
		\caption{The measurement covariance (or variance) tracking performance of VBKF and STKF-AR in Case 3. (a) The performance of VBKF. (b) The performance of STKF-AR.}
		\label{case3}
	\end{figure}
	\subsection{Example 4}
	We consider the proprioceptive quadruped localization problem by fusing robot kinematics and IMU measurements on slippery terrain~\cite{li2026interacting}. The system state at discrete time step $k$ is defined as:$$x_k = \{R_k, p_k, v_k, b^g_k, b^a_k, d_{0,k}, d_{1,k}, d_{2,k}, d_{3,k}\}$$where $R_k \in SO(3)$ denotes the base orientation; $p_k, v_k \in \mathbb{R}^3$ represent the base position and velocity, respectively; $b^g_k, b^a_k \in \mathbb{R}^3$ are the gyroscope and accelerometer biases; and $d_{i,k} \in \mathbb{R}^3$ (for $i \in \{0, 1, 2, 3\}$) represents the position of the $i$-th foot. The system dynamics are governed by
	\begin{equation}
		\begin{aligned}
			R_{k+1} &= R_k \exp\big( (\omega_k - b^g_k - \eta^g_k)_\times \Delta t \big), \\
			v_{k+1} &= v_k + \big( R_k (a_k - b^a_k - \eta^a_k) + g \big) \Delta t + \eta_{v,k}, \\
			b^g_{k+1} &= b^g_k + \eta_{bg,k}, \\
			b^a_{k+1} &= b^a_k + \eta_{ba,k},\\
			p_{k+1} &= p_k + v_k \Delta t + \eta_{p,k}, \\
			d_{i,k+1} &= d_{i,k} + v_{di,k} \Delta t + \eta_{di,k}
		\end{aligned}
	\end{equation}
	where $\omega_k$ and $a_k$ are the measured angular velocity and linear acceleration, respectively; $g$ is the gravity vector; $v_{di,k}$ is the foot velocity; $(\cdot)_\times$ denotes the skew-symmetric matrix operator; and $\eta^g_k, \eta^a_k, \eta_{bg,k}, \eta_{ba,k}, \eta_{v,k}, \eta_{p,k}, \eta_{di,k}$ represent mutually independent, zero-mean Gaussian white noise processes driving the respective states; $\delta=0.001$ s is the sampling time and the whole interval is $10$ s. The foot position measurements derived from the joint encoders are modeled as:$$d_{i,k} = f_{\text{kin},i}(q_k) + \epsilon_{i,k}, \quad i \in \{0, \dots, 3\}$$where $f_{\text{kin},i}(\cdot)$ denotes the forward kinematics function for the $i$-th leg, $q_k$ represents the joint configurations, and $\epsilon_{i,k}$ is the associated zero-mean Gaussian measurement noise.
	
	The propagation of the foot positions exhibits a contact-dependent switched behavior. During the swing phase, the foot velocity $v_{di,k}$ is unobservable. Consequently, we model it as $v_{di,k} \triangleq 0$ and assign an infinite process noise covariance ($Q_{di} \to \infty$) to reflect this uncertainty. During the stance phase, assuming a rolling contact model, the foot velocity is determined by $v_{di,k} = (\omega_{i,k})_{\times} r_k$. Here, $\omega_{i,k}$ is the angular velocity of the foot derived from the kinematics, and $r_k = [0, 0, r_c]^T$ is the vector from the ground contact point to the joint center, with $r_c$ being the effective rolling radius. Note that the Gaussian assumption for the process noise $\eta_{di,k}$ is violated in two primary scenarios: during the initial touchdown phase, where ground impact induces impulsive disturbances, and when the rolling contact assumption fails due to foot slippage. 
	
	We simulate a Unitree A1 quadruped traversing flat terrain interspersed with slippery belts in Gazebo, adopting the environment design from \cite{liao2023walking}. To address the impulsive process noise and potential slippage, we set $\nu_p = [10^8 \mathbf{1}_{15\times 1}^T, \mathbf{1}_{12\times 1}^T]^T$ and $\nu_r = 10^8 \mathbf{1}_{12\times 1}$. Meanwhile, we apply $\rho_p = [\mathbf{1}_{15\times 1}^T, \rho_d^T, \rho_d^T, \rho_d^T, \rho_d^T]^T$, where $\rho_d = [0.9, 0.9, 1]^T$, and $\rho_r = \mathbf{1}_{12\times 1}$. For the STKF baseline, we retain the same $\nu_p$ and $\nu_r$, but assign $\rho_p = \mathbf{1}_{27\times 1}$ and $\rho_r = \mathbf{1}_{12\times 1}$. We benchmark the proposed method against a conventional Error-State Kalman Filter (ESKF) \cite{li2026interacting}, and an Invariant Extended Kalman Filter (IEKF)~\cite{hartley2020contact}. The localization errors are detailed in Table~\ref{tab:rmse_comparison}, where the RMSEs of a vector $e_p$ is computed via $e_p=\|[e_{px},e_{py},e_{pz}]\|_2$ with $e_{px}$, $e_{py}$, and $e_{pz}$ are the RMSE of each axis. The corresponding trajectory visualizations are in Fig.~\ref{trajectory}. These results demonstrate the effectiveness of the proposed methods.
	\begin{table}[htbp]
		\centering
		\caption{RMSEs of the base position, velocity, and Euler angle of different estimators.}
		\label{tab:rmse_comparison}
		\scalebox{0.65}{
			\begin{tabular}{c c c c}
				\hline
				\hline
				\textbf{Alg} & \textbf{$e_p$ (m)} & \textbf{$e_v$ (m/s)} & \textbf{$e_\theta$ (deg)} \\
				\hline
				IEKF    & 0.0467 & 0.0383 & 0.4450 \\
				ESKF    & 0.1117 & 0.0382 & 0.5097 \\
				STKF    & 0.0348 & 0.0230 & 0.2477 \\
				STKF-AR & 0.0247 & 0.0245 & 0.2362 \\
				\hline
				\hline
		\end{tabular}}
	\end{table}
	\begin{figure}[htp]
		\centering
		\includegraphics[width=0.85\columnwidth]{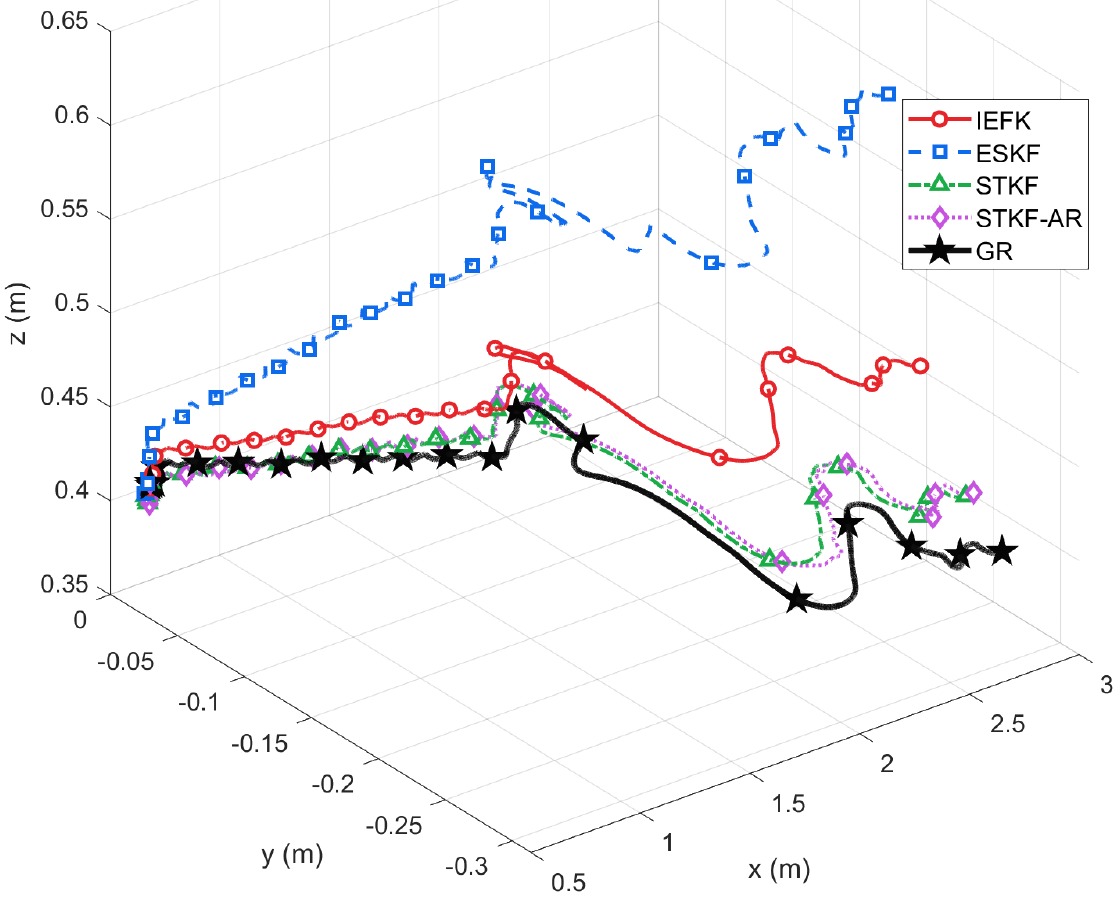}
		\caption{3D trajectory visualization of different estimators. The GR denotes the ground truth position.}
		\label{trajectory}
	\end{figure}
	\section{Conclusion}
	This work bridges the gap between the robust Kalman filter and the adaptive filter. Specifically, we prove that the STKF, derived by the Student's $t$-distribution induced loss and solved by fixed-point iteration, can be understood as a prerequisite of the VBKF. On this basis, we provide necessary conditions for a class of losses that can be solved by a fixed-point solution, which is much more computation-efficient than gradient-based solutions. Leveraging the variational technique, we derive two robust-adaptive filters, STKF-A and STKF-AR. We demonstrate that there is a trade-off between tracking speed and tracking variance in terms of covariance tracking in adaptive filters, highlighting the importance of selecting proper forgetting factors. Our proposed approaches can recover KF, STKF (robust filters), VBKF (adaptive filters), and can address complex noise scenarios with mixing outliers and adaptive noises.  Simulations verify the effectiveness of the proposed method. 
	
	\section{Appendix}
	\subsection{Appendix A}
	\label{proofproperty2}
	\begin{pf}
		According to $\lim \limits_{x\to 0}\log (1+x) =x$, it follows that
		\begin{equation}
			\lim \limits_{\nu\to \infty} \frac{\nu}{2}\log\left(1+ \frac{e^2}{\nu\tau^2}\right)=\frac{1}{2}\frac{e^2}{\tau^2}
		\end{equation}
		and hence $\mathcal{L}_{st}=\mathcal{L}_{gau}$ as $\nu \to \infty$. As $\nu=1$, one observes that $\mathcal{L}_{st}=\frac{1}{2}\log\left(1+\frac{e^2}{\tau^2}\right)$ which becomes Cauchy loss and corresponds to Cauchy distribution. This completes the proof.
	\end{pf}
	\subsection{Appendix B}
	\begin{pf}
		\label{proofproperty3}
		It is easy to obtain the Hessian matrix of $\mathcal{L}_{st}$ as follows
		\begin{equation}
			H({\mathcal{L}_{st}})=\frac{\partial^2 \mathcal{L}_{st}}{\partial e^2} = \frac{\nu(\nu\tau^2-e^2)}{(\nu\tau^2+e^2)^{2}}.
		\end{equation}
		It is obvious that $H({\mathcal{L}_{st}}) \ge 0$ as $e \in [-\sqrt{\nu}\tau,\sqrt{\nu}\tau]$ and $H({\mathcal{L}_{st}}) < 0$ otherwise. This completes the proof.
	\end{pf}
	\subsection{Appendix C}
	\begin{pf}
		\label{proofproperty4}
		According to \eqref{plambda}, one has 
		$$p(\lambda)=\frac{(\nu\tau^2/2)^{\nu/2}}{\Gamma(\nu/2)}(1 / \lambda)^{\nu/2+1} \exp \left(-\frac{\nu\tau^2}{2\lambda}\right).$$
		By applying Stirling's approximation, we have $\Gamma(\nu/2) \sim \sqrt{2\pi}(\frac{\nu}{2})^{\nu/2-0.5}\exp(-\nu/2)$. It follows that
		\begin{equation}
			\scalebox{0.9}{$\displaystyle
				\begin{aligned}
					\lim \limits_{\nu \to \infty}p(\lambda) &=\frac{(\nu\tau^2/2)^{\nu/2}}{\sqrt{2\pi}(\frac{\nu}{2})^{\nu/2-0.5}\exp(-\nu/2)} (1 / \lambda)^{\nu/2+1} \exp \left(-\frac{\nu\tau^2}{2\lambda}\right)\\
					&=\sqrt{\frac{\nu}{4\pi}} \frac{1}{\lambda} \left( \frac{\tau^2}{\lambda} \exp\left(1 - \frac{\tau^2}{\lambda}\right) \right)^{\nu/2}.
				\end{aligned}
				$}
		\end{equation}
		Let $f(\lambda) = \frac{\tau^2}{\lambda} \exp(1 - \frac{\tau^2}{\lambda})$. It is easy to verify that $f(\lambda)$ reaches its unique global maximum at $\lambda = \tau^2$, where $f(\tau^2) = 1$. For any $\lambda \neq \tau^2$, we have $0 < f(\lambda) < 1$. Consequently, as $\nu \to \infty$, the term $f(\lambda)^{\nu/2}$ approaches $0$ everywhere except at $\lambda = \tau^2$. Meanwhile, the prefactor $\sqrt{\frac{\nu}{4\pi}} \frac{1}{\lambda}$ goes to $\infty$ at $\lambda = \tau^2$, ensuring the integral over the probability density space remains $1$. This indicates that $p(\lambda)$ converges to a shifted Dirac delta function $\delta(\lambda-\tau^2)$. This completes the proof.
	\end{pf}
	\subsection{Appendix D}
	\label{pftheorem1}
	\begin{pf}
		Denoting $R_{ww}=W_k^{T}D_k W_k$ and $R_{wt}=W_k^{T}D_k t_k$, it follows that $f(x)=R_{ww}^{-1}R_{wt}$ according to \eqref{fixedpoint}. Since the induced norm is compatible with vector $\ell_p$ norm, we have
		\begin{equation}
			\|f(x)\|_1=\|R_{ww}^{-1}R_{wt}\|_1 \le \|R_{ww}^{-1}\|_1\|R_{wt}\|_1.
			\label{fixp}
		\end{equation}
		According to matrix theory, we have
		\begin{equation}
			\tiny
			\begin{aligned}
				\|R_{ww}^{-1}\|_1 &\le \sqrt{n} \|R_{ww}^{-1}\|_2 = \sqrt{n}\lambda_{\max}[R_{ww}^{-1}]= \frac{\sqrt{n}}{\lambda_{\min}[R_{ww}]} \\
				&=\frac{\sqrt{n}}{\lambda_{\min}[\sum_{i=1}^{l}w_i^{T}d_{\nu_i}(e_i) w_i]} \\
				&\overset{(\mathrm{I})}{\le}  \frac{\sqrt{n}}{\lambda_{\min}[\sum_{i=1}^{l}w_i^{T}d_{\nu_i}(|t_i|+\gamma |w_i^{T}|_1) w_i]},
			\end{aligned}
			\label{RWW}
		\end{equation}
		where (I) comes from $|e_i|=|t_i-w_i x| \le |t_i|+ \gamma \|w_i\|_1$. Furthermore, it holds that 
		\begin{equation}
			\begin{aligned}
				\|R_{wt}\|_1&=\|\sum_{i}^{l} w_i^{T}d_{\nu_i}(e_i)t_i\|_1  \overset{(\mathrm{II})}{\le} \sum_{i}^{l}\frac{1}{\tau_i^2}\|w_i^{T}\|_1 |t_i|,
			\end{aligned}
			\label{rwt}
		\end{equation}
		where $(\mathrm{II})$ comes form the fact $d_{\nu_i}(e_i) \le \frac{1}{\tau_i^2}$. Substituting the expression of $R_{ww}$ and $R_{wt}$ into \eqref{fixp}, one arrives
		\begin{equation}
			\|f(x)\|_{1} \le \boldsymbol{\phi}(\boldsymbol{\nu})= \frac{\sqrt{n}\sum_{i}^{l}\frac{1}{\tau_i^2}\|w_i^{T}\|_1 |t_i| }{\lambda_{\min}[\sum_{i=1}^{l}w_i^{T}d_{\nu_i}(|t_i|+\gamma |w_i^{T}|_1) w_i]}.
			\label{fxb}
		\end{equation}
		In the case $\nu_1=\nu_2=\ldots=\nu_{l}=\nu$, $\boldsymbol{\phi}(\boldsymbol{\nu})$ degenerates to  $\phi(\nu)$ as follows:
		\begin{equation}
			\phi(\nu)= \frac{\sqrt{n}\sum_{i=1}^{l}|t_i|\|w_i^{T}\|_1}{\lambda_{\min}\Big{[}\sum_{i=1}^{l} d_{\nu}(\gamma \|w_i^{T}\|_1+t_i)w_i^{T}w_i\Big{]}},
		\end{equation} 
		which is a \emph{decreasing function} of $\nu$. Moreover, we have 
		$$
		\lim \limits_{\nu \to \infty} \phi(\nu) = \xi = \frac{\sqrt{n}\sum_{i}^{l}\frac{1}{\tau_i^2}\|w_i^{T}\|_1 |t_i| }{\lambda_{\min}[\sum_{i=1}^{l} \frac{1}{\tau_i^2} w_i^{T} w_i]}, 	\lim \limits_{\nu \to 0_{+}} \phi(\nu) = \infty.
		$$
		This indicates that if $\gamma > \xi$, $\phi(\nu) = \gamma $ always has a unique solution $\nu^{*}$ over $[0,\infty]$. Subsequently, we consider a much more general case $\nu_i \ge \nu^{*}$ for $i=1,2,\ldots,l$. One observes that $w_i^{T}d_{\nu_i}(|t_i|+\gamma |w_i^{T}|_1) w_i$ is a positive diagonal matrix and $d_{\nu_i}(\cdot) \ge d_{\nu^{*}}(\cdot)$ if $\nu_i \ge \nu^{*}$. It follows that 
		$$
		\|f(x)\|_{1} \le \phi(\boldsymbol{\nu}) \le \phi(\nu^{*})= \gamma.
		$$
		This completes the proof.
	\end{pf}
	\subsection{Appendix E}
	\label{pftheorem2}
	\begin{pf}
		To prove $\|\nabla_{x} f(x)\| \le \eta$, it is sufficient to prove $\forall j$, $\|\frac{\partial }{\partial x_j}\|_{1} \le \eta$. Based on the fact that $$\frac{\partial \mathbf{U}^{-1}}{\partial \mathrm{x}}=-\mathbf{U}^{-1}\frac{\partial \mathbf{U}}{\partial \mathrm{x}}\mathbf{U}^{-1},~\frac{\partial \mathbf{U}\mathbf{V}}{\partial \mathrm{x}}= \frac{\partial \mathbf{U}}{\partial\mathrm{x}}\mathbf{V}+\mathbf{U}\frac{\partial \mathbf{V}}{\partial\mathrm{x}}$$ where $\mathbf{U}$ and $\mathbf{V}$ are matrices and $\mathrm{x}$ is a scalar, we have the following equation:
		\begin{equation}
			\begin{aligned}			\frac{\partial }{\partial {x}_j} f(x)&=\frac{\partial}{{x}_j}R_{ww}^{-1} R_{wt}\\
				&= -R_{ww}^{-1}\Big{[} \sum_{i=1}^{l}\frac{2 \nu_i w_{i,j} e_i}{(\nu_i \tau_i^{2}+e_i^{T}e_i)^{2}}{w}_i^{T}{w}_i\Big{]} f(x)\\
				&+R_{ww}^{-1} \Big{[}\sum_{i=1}^{l} \frac{2 \nu_i w_{i,j} e_i}{(\nu_i \tau_i^{2}+e_i^{T}e_i)^{2}}{w}_i^{T}{t}_i\Big{]},
				\label{gradient}
			\end{aligned}
		\end{equation}
		where $w_{i,j}$ is the $j$-th element of $w_i$ and $x_j$ is $j$-th element of vector $x$. Take one norm in both sides of \eqref{gradient}, we have 
		\begin{equation}
			\begin{aligned}
				\|\frac{\partial }{\partial x_j} f(x) \|_{1} 
				&\le     
				\Big{\|}-R_{ww}^{-1}\Big{[} \sum_{i=1}^{l}\frac{2 \nu_i w_{i,j} e_i}{(\nu_i \tau_i^{2}+e_i^{T}e_i)^{2}}{w}_i^{T}{w}_i\Big{]} f(x)\Big{\|_1} \\
				& + \Big{\|}R_{ww}^{-1} \Big{[}\sum_{i=1}^{l} \frac{2 \nu_i w_{i,j} e_i}{(\nu_i \tau_i^{2}+e_i^{T}e_i)^{2}}{w}_i^{T}{t}_i\Big{]} \Big{\|_1}.
				\label{gradient1}
			\end{aligned}
		\end{equation}
		For the first term on the right side of \eqref{gradient1}, we have
		\begin{equation}
			\tiny
			\begin{aligned}
				&\Big{\|}-R_{ww}^{-1}\Big{[} \sum_{i=1}^{l}\frac{2 \nu_i w_{i,j} e_i}{(\nu_i \tau_i^{2}+e_i^{T}e_i)^{2}}{w}_i^{T}{w}_i\Big{]} f(x)\Big{\|_1}  \\
				&\le 2\|R_{ww}^{-1}\|_1 \Big{\|}\Big{[} \sum_{i=1}^{l}\frac{\nu_i w_{i,j} e_i}{(\nu_i \tau_i^{2}+e_i^{T}e_i)^{2}}{w}_i^{T}{w}_i\Big{]} \Big{\|}_1 \|f(x)\|_1\\
				& \overset{(\mathrm{I})}{\le}2\gamma\|R_{ww}^{-1}\|_1  \sum_{i=1}^{l}\Big{\|}\frac{\nu_i w_{i,j} e_i}{(\nu_i \tau_i^{2}+e_i^{T}e_i)^{2}} {w}_i^{T}{w}_i \Big{\|}_1 \\
				&\overset{(\mathrm{II})}{\le} 2\gamma\|R_{ww}^{-1}\|_1   \sum_{i=1}^{l}\frac{|t_i|+\gamma\|w_i^{T}\|_{1}}{\nu_i \tau_i^4}\|w_i^{T}\|_1\|w_i^{T}w_i\|_1,
				\label{gra_part1}
			\end{aligned}
		\end{equation}
		where (I) comes from the convexity of vector $\ell_1$ norm and $f(x) \le \gamma$, and (II) comes from $|w_{i,j}e_i| \le (|t_i|+\gamma \|w_i\|_1 )\|w_i^{T}\|_1$ and $\frac{\nu_i}{(\nu_i \tau_i^{2}+e_i^{T}e_i)^{2}} \le \frac{1}{\nu_i \tau_i^{4}}$.
		Similarly, we have
		\begin{equation}
			\tiny
			\begin{aligned}
				&\Big{\|}R_{ww}^{-1} \Big{[}\sum_{i=1}^{l} \frac{2 \nu_i w_{i,j} e_i}{(\nu_i \tau_i^{2}+e_i^{T}e_i)^{2}}{w}_i^{T}{t}_i\Big{]} \Big{\|_1} \\
				&\le 2\|R_{ww}^{-1}\|_{1} \sum_{i=1}^{l} \frac{|t_i|+\gamma\|w_i^{T}\|_{1}}{\nu_i \tau_i^4} \|w_i^{T}\|_1\|w_i^{T}t_i\|_1.
				\label{gra_part2}
			\end{aligned}
		\end{equation}
		Substituting \eqref{RWW}, \eqref{gra_part1}, and \eqref{gra_part2} into \eqref{gradient1}, we obtain
		\begin{equation}
			\begin{aligned}
				&\|\frac{\partial }{\partial x_j} f(x) \|_{1} \le \boldsymbol{\psi}({\boldsymbol{\nu}})\\&= \frac{2\sqrt{n} \sum_{i=1}^{l}\frac{|t_i|+\gamma\|w_i^{T}\|_{1}}{\nu_i \tau_i^{4}} \|w_i^{T}\|_1  \big{(}\gamma\|w_i^{T}w_i \|_1+ \|w_i^{T}t_i\|_1 \big{)} }{\lambda_{\min}[\sum_{i=1}^{l}w_i^{T}d_{\nu_i}(|t_i|+\gamma |w_i^{T}|_1) w_i]}.
				\label{psibar}
			\end{aligned}
		\end{equation}
		By setting $\nu_i=\nu$ for all $i$, the function $\boldsymbol{\psi}({\boldsymbol{\nu}})$ degenerates to \eqref{con2}. In such cases, one has
		\begin{equation}
			\begin{aligned}
				&\|\frac{\partial }{\partial x_j} f(x) \|_{1} \le {\psi}({\nu}).
				\label{psi}
			\end{aligned}
		\end{equation}
		One can see that \eqref{psi} is a continuous and strictly decreasing function satisfying $\lim\limits_{\nu \to 0^{+}} \psi(\nu) = \infty$ and $\lim\limits_{\nu \to \infty} \psi(\nu) = 0$. This implies that $\psi(\nu)=\eta$ has a unique solution $\nu^{+}$ and $\psi(\nu) \le \eta$ if $
		\nu \ge
		\nu^{+}$. Observing \eqref{psibar} and \eqref{psi}, we have $\boldsymbol{\psi}(\boldsymbol{\nu}) \le \psi(\nu^{+})$ if $\nu_i \ge \nu^{+}$ for all $i$. This indicates that $0<\boldsymbol{\psi}(\boldsymbol{\nu}) \le \eta$ if $\forall i, \nu_i \ge \nu^{+}$. This completes the proof.
	\end{pf}
	\subsection{Appendix F}
	\label{pftheorem3}
	\begin{pf}
		One observes that \emph{Condition 1} ensures that $J_{\boldsymbol{\nu}}(e)$ is a well-defined loss function. \emph{Condition 2} guarantees the validity of \eqref{fxb}. \emph{Condition 3} guarantees the validity of \eqref{gra_part2}. \emph{Condition 4} (combined with the above three conditions) guarantees the hold of the inequality in \eqref{convergence} by analogy with the proof of Theorems \ref{theorem1} and \ref{theorem2}. This completes the proof.
	\end{pf}
	\subsection{Appendix G}
	\label{vbderivation}
	The system dynamics, according to \eqref{linreg}, are as follows:
	\begin{equation}
		\begin{aligned}
			t_k = W_k x_k + \zeta_k.
		\end{aligned}
		\label{linreg1}
	\end{equation}
	We follow the standard VB-approach and approximate the joint distribution using a factored form as follows:
	\begin{equation}
		p(x_k, R_{\zeta\zeta_k} | t_k) \approx Q_x(x_k) Q_R(R_{\zeta\zeta_k}).
	\end{equation}
	Then, one can minimize the Kullback-Leibler (KL) divergence between the approximated distribution and the true posterior:
	\begin{equation}
		\begin{aligned}
			&\operatorname{KL}(Q_x Q_R || p(x_k, R_{\zeta\zeta_k} | t_k)) \\
			&= \iint Q_x Q_R \log \left( \frac{Q_x Q_R}{p(x_k, R_{\zeta\zeta_k} | t_k)} \right) \mathrm{d} x_k \mathrm{d} R_{\zeta\zeta_k}.
		\end{aligned}
	\end{equation}
	Minimizing the KL divergence with respect to  $Q_R(R_{\zeta\zeta_k})$ gives
	\begin{equation}
		\begin{aligned}
			Q_R(R_{\zeta\zeta_k}) &\propto  \exp (\mathbb{E}_{Q_x}[\log p(x_k, R_{\zeta\zeta_k}, t_k)]).
		\end{aligned}
	\end{equation}
	Letting $\lambda_{i,k}$ denote the $i$-th diagonal element of $R_{\zeta\zeta_k}$, and dropping terms independent of $R_{\zeta\zeta_k}$, one has:
	\begin{equation}
		\tiny
		\begin{aligned}
			\mathbb{E}_{Q_x}&[\log p(x_k, R_{\zeta\zeta_k}, t_k)] \propto \mathbb{E}_{Q_x}[\log p(t_k | x_k, R_{\zeta\zeta_k})] + \log p(R_{\zeta\zeta_k})\\
			&= -\sum_{i=1}^{l} \left( \frac{\nu_{i,k}^- + 1}{2} + 1 \right) \ln (\lambda_{i,k}) \\
			&\quad - \sum_{i=1}^{l} \frac{1}{\lambda_{i,k}} \left( \frac{\nu_{i,k}^- (\tau_{i,k}^2)^-}{2} + \frac{1}{2} \langle (t_k - W_k x_k)_i^2 \rangle_x \right) ,
		\end{aligned}
		\label{ex}
	\end{equation}
	where $\langle \cdot \rangle_x = \int (\cdot) Q_x(x_k) \mathrm{d} x_k$. By evaluating the expectations in \eqref{ex}, and matching the parameters of the distributions on the left and right hand sides to the inverse gamma distribution, we obtain
	\begin{equation}
		\begin{aligned}
			Q_R(R_{\zeta\zeta_k}) &= \prod_{i=1}^{l} \operatorname{Inv-Gam} \left( \lambda_{i,k} \bigg| \frac{\nu_{i,k}^+}{2}, \frac{\nu_{i,k}^+ (\tau_{i,k}^2)^+}{2} \right),
		\end{aligned}
		\label{des}
	\end{equation}
	where $\hat{x}_k$ and $P_k$ are obtained by the fixed-point iteration by solving \eqref{fixedpoint}, and the updated Inverse-Gamma parameters $\nu_{i,k}^+$ and $(\tau_{i,k}^2)^+$ are given by the following equations:
	\begin{equation}
		\begin{aligned}
			\nu_{i,k}^{+} &= \nu_{i,k}^{-} + 1, \\
			(\tau_{i,k}^2)^+ &= \frac{\nu_{i,k}^-}{\nu_{i,k}^+} (\tau_{i,k}^2)^{-} + \frac{e_{i,k}^2 + [W_k P_k W_k^{T}]_{ii}}{\nu_{i,k}^{+}}.
		\end{aligned}
	\end{equation}
	This completes the proof.
	\subsection{Appendix H}
	\label{appenB}
	According to the inverse-Gamma distribution in \eqref{plambda}, we have
	\begin{equation}
		\ln p(\lambda) = -\left(\frac{\nu}{2}+1\right)\ln\lambda - \frac{\nu\tau^2}{2\lambda} + C,
	\end{equation}
	where $C$ is a constant independent of $\lambda$. The mode $\lambda_0$ is obtained by setting the first derivative to zero:
	\begin{equation}
		\left.\frac{d}{d \lambda} \ln p(\lambda)\right|_{\lambda_0} = -\frac{\nu+2}{2\lambda_0} + \frac{\nu\tau^2}{2\lambda_0^2} = 0 \implies \lambda_0 = \frac{\nu\tau^2}{\nu+2}.
	\end{equation}
	Because $\lambda_0$ is a local maximum, the first-order term in the Taylor expansion vanishes exactly. The second derivative evaluated at the mode is:
	\begin{equation}
		\left.\frac{d^2}{d \lambda^2} \ln p(\lambda)\right|_{\lambda_0} = \frac{\nu+2}{2\lambda_0^2} - \frac{\nu\tau^2}{\lambda_0^3} = -\frac{(\nu+2)^3}{2\nu^2\tau^4}.
	\end{equation}
	Consequently, the Taylor expansion of $\ln p(\lambda)$ simplifies to:
	\begin{equation}
		\ln p(\lambda) \approx \ln p(\lambda_0) - \frac{(\nu+2)^3}{4\nu^2\tau^4}(\lambda-\lambda_0)^2.
	\end{equation}
	By matching this quadratic form to the logarithm of a Gaussian probability density function $\ln q(\lambda;\mu,\sigma^2) = -\frac{(\lambda-\mu)^2}{2\sigma^2} + C^\prime$, we obtain the approximated normal distribution:
	\begin{equation}
		q(\lambda) \sim \mathcal{N}\left(\frac{\nu\tau^2}{\nu+2}, \frac{2\nu^2\tau^4}{(\nu+2)^3}\right).
	\end{equation}
	\subsection{Appendix I}
	\label{prooftheorem5}
	\begin{pf}
		In steady state ($k \to \infty$), the sequence converges to a weakly stationary random variable $\lambda_{i,\infty}^{+}$ satisfying:
		\begin{equation*}
			\lambda_{i,\infty}^{+} = \rho_i \lambda_{i,\infty}^{+} + (1-\rho_i)\mathcal{E}_{i,\infty}
		\end{equation*}
		
		\textbf{Mean:} Taking the expectation of both sides and substituting $\mathbb{E}[\mathcal{E}_{i,\infty}] = \sigma_i^2$ yields:
		\begin{equation*}
			\mathbb{E}[\lambda_{i,\infty}^{+}] = \rho_i \mathbb{E}[\lambda_{i,\infty}^{+}] + (1-\rho_i)\sigma_i^2 \implies \mathbb{E}[\lambda_{i,\infty}^{+}] = \sigma_i^2.
		\end{equation*}
		
		\textbf{Variance:} Because $\lambda_{i,k-1}^{+}$ strictly depends on past errors, it is uncorrelated with the current error $\mathcal{E}_{i,k}$. Taking the variance of both sides gives:
		\begin{equation*}
			\mathrm{Var}(\lambda_{i,\infty}^{+}) = \rho_i^2 \mathrm{Var}(\lambda_{i,\infty}^{+}) + (1-\rho_i)^2 \mathrm{Var}(\mathcal{E}_{i,\infty}).
		\end{equation*}
		Substituting $\mathrm{Var}(\mathcal{E}_{i,\infty}) = \eta_i^4$ and isolating $\mathrm{Var}(\lambda_{i,\infty}^{+})$, we obtain:
		\begin{equation*}
			\mathrm{Var}(\lambda_{i,\infty}^{+}) = \frac{(1-\rho_i)^2}{1-\rho_i^2} \eta_i^4 = \frac{1-\rho_i}{1+\rho_i} \eta_i^4.
		\end{equation*}
	\end{pf}
	
	\bibliographystyle{IEEEtranN}
	\bibliography{reference}
\end{document}